\documentclass{sig-alternate}

\usepackage{amsmath}
\usepackage{amssymb}
\usepackage{color}
\usepackage[usenames,dvipsnames]{xcolor}
\usepackage{graphicx}
\usepackage{units}
\usepackage{multirow}

\usepackage{algpseudocode}
\usepackage{algorithm}

\newcommand{\hull}{\mathrm{hull}}
\newcommand{\vertices}{V}

\newcommand{\supp}{\mathrm{Supp}}
\newcommand{\dest}{\mathrm{Succ}}

\newcommand{\post}{\mathrm{Post}}
\newcommand{\preo}{\mathrm{Pre}}
\newcommand{\prer}{\mathrm{PreR}}
\newcommand{\prep}{\mathrm{PreP}}
\newcommand{\attr}{\mathrm{Attr}}
\newcommand{\attrr}{\mathrm{AttrR}}

\newcommand{\out}{\mathsf{out}}

\newcommand{\F}{\mathbf{F}}

\newcommand{\T}{\mathcal{T}}
\newcommand{\X}{\mathcal{X}}
\newcommand{\U}{\mathcal{U}}
\newcommand{\W}{\mathcal{W}}
\newcommand{\N}{\mathcal{N}}
\newcommand{\Z}{\mathcal{Z}}

\newcommand{\G}{\mathcal{G}}
\newcommand{\product}{\mathcal{P}}

\newcommand{\act}{Act}

\newcommand{\ie}{{\it i.e., }}
\newcommand{\eg}{{\it e.g., }}

\newtheorem{remark}{Remark}
\newtheorem{definition}{Definition}
\newtheorem{problem}{Problem}
\newtheorem{theorem}{Theorem}
\newtheorem{proposition}{Proposition}
\newtheorem{corollary}{Corollary}
\newtheorem{example}{Example}

\newcommand{\rev}[2]{{\color{Magenta}#1}}

\graphicspath{{./fig/}}

\makeatletter
\def\@copyrightspace{\relax}
\makeatother

\advance\textwidth5mm
\advance\hoffset-2.5mm
\advance\textheight5mm
\advance\voffset-2.5mm

\begin{document}

\title{Temporal Logic Control for Stochastic Linear Systems using Abstraction Refinement of Probabilistic Games
\titlenote{
{\scriptsize This work was partially supported by Czech Science Foundation grant 15-17564S, People Programme (Marie Curie Actions) of the European Union's Seventh Framework Programme (FP7/2007-2013) under REA grant agreement 291734, ERC grant 267989 (QUAREM) and Start grant (279307: Graph Games), Austrian Science Fund (FWF) grants S11402-N23 (RiSE), P23499-N23 and S11407-N23 (RiSE), Czech Ministry of Education Youth and Sports grant LH11065, and NSF grants CMMI-1400167 and CNS-1035588.}
}}
%
%
%
%
%

\numberofauthors{6} 
%
\author{
%
%
\alignauthor
M\'{a}ria Svore\v{n}ov\'{a}
\\
       \affaddr{Faculty of Informatics}\\
       \affaddr{Masaryk University}\\
       \affaddr{Brno, Czech republic}\\
       \email{svorenova@mail.muni.cz}
\alignauthor Jan K\v{r}et\'{i}nsk\'{y} 
\\
       \affaddr{IST Austria}\\
       \affaddr{Klosterneuburg, Austria}\\
       \email{jan.kretinsky@ist.ac.at}
\alignauthor
Martin Chmel\'{i}k
\\
       \affaddr{IST Austria}\\
       \affaddr{Klosterneuburg, Austria}\\
       \email{martin.chmelik@ist.ac.at}
\and  
\alignauthor Krishnendu Chatterjee\\
       \affaddr{IST Austria}\\
       \affaddr{Klosterneuburg, Austria}\\
       \email{kchatterjee@ist.ac.at}
\alignauthor Ivana \v{C}ern\'{a}\\
       \affaddr{Faculty of Informatics}\\
       \affaddr{Masaryk University}\\
       \affaddr{Brno, Czech republic}\\
       \email{cerna@muni.cz}
\alignauthor Calin Belta\\
       \affaddr{Dep. of Mechanical Eng.}\\
       \affaddr{Boston University}\\
       \affaddr{Boston, MA, USA}\\
       \email{cbelta@bu.edu}
}

\maketitle

\begin{abstract}
We consider the problem of computing the set of initial states of a dynamical system such that there exists a control strategy to ensure that the trajectories satisfy a temporal logic specification with probability 1 (almost-surely). We focus on discrete-time, stochastic linear dynamics and specifications given as formulas of the Generalized Reactivity(1) fragment of Linear Temporal Logic over linear predicates in the states of the system. We propose a solution based on iterative abstraction-refinement, and turn-based 2-player probabilistic games. While the theoretical guarantee of our algorithm after any finite number of iterations is only a partial solution, we show that if our algorithm terminates, then the result is the set of satisfying initial states. Moreover, for any (partial) solution our algorithm synthesizes witness control strategies to ensure almost-sure satisfaction of the temporal logic specification. We demonstrate our approach on an illustrative case study.
\end{abstract}






\section{Introduction}\label{sec:introduction}
The formal verification problem, in which the goal is to check whether 
behaviors of a finite model satisfy a correctness specification, received a lot of attention during the past thirty years \cite{Clarke99,baierbook}. 
In contrast, in the synthesis problem the goal is to synthesize or control a finite system from a temporal logic specification.
While the synthesis problem also has a long tradition \cite{church57,buchiL69,PR89}, it has gained significant attention in formal methods more recently. For example, these techniques are being deployed in control and path planning in particular: 
model checking techniques can be adapted to synthesize (optimal) controllers for deterministic finite systems \cite{SvCeBe-ACC-2013,mengicra13}, B\"uchi and Rabin games can be reformulated as control strategies for nondeterministic systems \cite{boyan,ufukhscc12}, and probabilistic games can be used to compute controller for finite probabilistic systems such as Markov decision processes \cite{majacdc13,norman_fmsd10}.

With the widespread integration of physical and digital components in cyber physical systems, and the safety and security requirements in such systems, there is an increased need for the development of formal methods techniques for systems with infinite state spaces, normally modeled as difference or differential equations. Most of the works in the area use partitions and simulation / bisimulation relations to construct a finite abstraction of the system, followed by verification or control of the abstraction. Existing results showed that such approaches are feasible for discrete and continuous time linear systems \cite{TP03,AyDiLaBe-CDC-2012}. With some added conservatism, more complicated dynamics and stochastic dynamics can also be handled \cite{Girard:2010,Agung:ACC:2006}.

In this work, we focus on the problem of finding the set of initial states of a dynamic system from which a given constraint can be satisfied, and synthesizing the corresponding witness control strategies. In particular, we consider discrete-time continuous-domain linear stochastic dynamics with the constraints given as formulas of the Generalized Reactivity(1) (GR(1))~\cite{gr1def} fragment of Linear Temporal Logic (LTL) over linear predicates in the states of the system. The GR(1) fragment offers polynomial computational complexity as compared to the doubly exponential one of general LTL, while being expressive enough to describe most of the usually considered temporal properties \cite{gr1def}. We require the formula to be satisfied almost-surely, \ie with probability 1. The almost-sure satisfaction is the strongest probability guarantee one can achieve while accounting for the stochasticity of the dynamics.


In our proposed approach, we iteratively construct and refine a discrete abstraction of the system and solve the synthesis problem for the abstract model. The discrete model considered in this work is a turn-based 2-player probabilistic game, also called 2\nicefrac{1}{2}-player game \cite{twoandhalfplayergames}. Every iteration of our algorithm produces a partial solution given as a partition of the state space into three categories. The first is a set of satisfying initial states together with corresponding witness strategies. The second is a set of non-satisfying initial states, \ie those from which the system cannot be controlled to satisfy the specification with probability 1. Finally, some parts of the state space may remain undecided due to coarse abstraction. As the abstraction gets more precise, more states are being decided with every iteration of the algorithm. The designed solution is partially correct. That means, we guarantee soundness, \ie almost sure satisfaction of the formula by all controlled trajectories starting in the satisfying initial set and non-existence of a satisfying control strategy for non-satisfying initial states. On the other hand, completeness is only ensured if the algorithm terminates. If a weaker abstraction model, such as 2 player games, was used, there would be no soundness guarantee on the non-satisfying initial states and no completeness guarantees.
 We provide a practical implementation of the algorithm that ends after a predefined number of iterations.

The main novelty of our work is the abstraction-refinement of a dynamic system using a 2\nicefrac{1}{2}-player game. While abstraction-refinement exists for discrete systems such as non-deterministic and probabilistic systems \cite{HJM03,CCD14,norman_fmsd10,CHJM05}, and some classes of hybrid systems \cite{norman_qest11, ozay_cdc14}, to the best of our knowledge, the approach that we present in this paper is the first attempt to construct abstraction-refinement of stochastic systems with continuous state and control spaces in the form of 2\nicefrac{1}{2}-player games. The game theoretic solutions are necessary to determine what needs to be refined, and the dynamics of the linear-stochastic systems determine the refinement steps. Thus both game theoretic aspects and the dynamics of the system play a crucial role in the refinement step, see Rem.~\ref{remark:interplay}.

This paper is closely related to \cite{boyan,ebruhscc12,mortezacdc12,abateTAC11,eric_cdc12}. Our computation of the abstraction is inspired from \cite{boyan}, which, however, does not consider stochastic dynamics and does not perform refinement. The latter issue is addressed in \cite{ebruhscc12} for non-stochastic dynamics and specifications with finite-time semantics in the form of syntactically co-safe LTL formulas. The exact problem that we formulate in this paper was also considered in \cite{mortezacdc12}, but for finite-time specifications in the form of probabilistic Computation Tree Logic (PCTL) formulas and for the particular case when the control space is finite. Also, in \cite{mortezacdc12}, the abstraction is constructed in the form of an interval-valued MDP, which is less expressive than the game considered here. An uncontrolled version of the abstraction problem for a stochastic system was considered in \cite{abateTAC11}, where the finite system was in form of a Markov set chain. In \cite{eric_cdc12}, the authors consider the problem of controlling uncertain MDPs from LTL specification. When restricted to almost sure satisfaction, uncertain MDPs have the same expressivity as the games considered here. To obtain a control strategy, the authors of \cite{eric_cdc12} use dynamic programming (value iteration), as opposed to games. 

The rest of the paper is organized as follows. We give some preliminaries in Sec.~\ref{sec:preliminaries} before we formulate the problem and outline the approach in Sec.~\ref{sec:pf}. The abstraction, game, and refinement algorithms are presented in Sec.~\ref{sec:solution}. A case study is included in Sec.~\ref{sec:casestudy}. We conclude with final remarks and directions for future work in Sec.~\ref{sec:futurework}. 

\section{Notation and preliminaries}\label{sec:preliminaries}

For a non-empty set $S$, let $S^{\omega}, S^*$ and $S^+$ 
denote the set of all infinite, finite and non-empty finite sequences of elements of $S$, respectively. For $\sigma\in S^+$ and $\rho\in S^\omega$, we use $|\sigma|$ to denote the length of 
$\sigma$, and $\sigma(n)$ and $\rho(n)$ to denote the $n$-th element, for $1\leq n\leq |\sigma|$ and $n\geq 1$, respectively. 
For two sets $S_1\subseteq S^*, S_2\subseteq S^*\cup S^\omega$, we use $S_1\cdot S_2=\{s_1\cdot s_2\mid s_1\in S_1,s_2\in S_2\}$ to denote their concatenation. Finally, for a finite set $S$, $|S|$ is the cardinality of $S$, $\mathcal{D}(S)$ is the set of all probability distributions over $S$ and $\{ s \in S \mid d(s)>0 \}$ is the support set of $d\in \mathcal{D}(S)$.

\subsection{Polytopes}\label{subsec:polytopes}
A (convex) polytope $\mathcal{P}\subset \mathbb{R}^N$ is defined as the convex hull of a finite set $X=\{x_i\}_{i\in I}\subset \mathbb{R}^N$, \ie 
\begin{equation}\label{eq:vrep}\small
\mathcal{P}= \hull(X) = \{\sum \limits_{i\in I} \lambda_i x_i \mid \forall i: \lambda_i\in [0, 1], \sum \limits_{i\in I}\lambda_i = 1\}.
\end{equation}
We use $\vertices(\mathcal{P})$ to denote the vertices of $\mathcal{P}$ that is the minimum set of vectors in $\mathbb{R}^N$ for which $\mathcal{P} = \hull(\vertices(\mathcal{P}))$. Alternatively, a polytope can be defined as an intersection of a finite number of half-spaces in $\mathbb{R}^N$, \ie
\begin{equation}\label{eq:hrep}
\mathcal{P}=\{x\in \mathbb{R}^N \mid H_{\mathcal{P}} x \leq K_{\mathcal{P}}\},
\end{equation}
where $H_{\mathcal{P}}, K_{\mathcal{P}}$ are matrices of appropriate sizes. Forms in Eq.~(\ref{eq:vrep}) and (\ref{eq:hrep}) are referred to as the V-representation and H-representation of polytope $\mathcal{P}$, respectively.  A polytope $\mathcal{P}\subset \mathbb{R}^N$ is called full-dimensional if it has at least $N+1$ vertices.  
In this work, we consider all polytopes to be full-dimensional, \ie if a polytope is not full-dimensional, we consider it empty. 

\subsection{Automata and Specifications}\label{subsec:gr1}
\newcommand{\run}{r}
\newcommand{\Inf}{\mathsf{Inf}}


\begin{definition}[$\omega$-automata]\label{def:ra}
A de\-ter\-ministic $\omega$-au\-to\-ma\-ton with B\"uchi implication (aka one-pair Streett) acceptance condition is a tuple $\mathcal{A}=(Q,\Sigma,\delta,q_{0},(E,F))$, where $Q$ is a non-empty finite set of states, $\Sigma$ is a finite alphabet, $\delta \colon Q\times \Sigma \to Q$ is a deterministic transition function, $q_{0}\in Q$ is an initial state, and $(E,F) \subseteq Q\times Q$ defines an acceptance condition. 
\end{definition}

Given an automaton, every word $w\in \Sigma^\omega$ over the alphabet $\Sigma$ induces a run which is an infinite sequence of states $q_0 q_1 \ldots\in Q^\omega$, such that $q_{i+1}= \delta(q_i,w(i))$ for all $i \geq 0$. Given a run $\run$, let $\Inf(\run)$ denote the set of states that appear infinitely often in $\run$. Given a B\"uchi implication acceptance condition $(E,F)$, a run $\run$ is accepting, if $\Inf(\run) \cap E \neq \emptyset$ implies $\Inf(\run) \cap F \neq \emptyset$, \ie if the set $E$ is visited infinitely often, then the set $F$ is visited infinitely often. The B\"uchi acceptance condition is a special case of B\"uchi implication acceptance condition where $E=Q$, \ie we require $F$ to be visited infinitely often. 
The language of an automaton is the set of words that induce an accepting run.





\begin{definition}[GR(1) formulae]\label{def:gr1}
A GR(1) formula $\varphi$ is a particular type of an LTL formula over alphabet $\Sigma$ of the form 
\begin{equation}\label{eq:gr1}
\varphi = \big( \bigwedge \limits_{i=1}^{m} \varphi_i \big) \, \Longrightarrow \, \big( \bigwedge \limits_{j=1}^{n} \varphi_j \big),
\end{equation}
where each $\varphi_i,\varphi_j$ is an LTL formula that can be represented by a deterministic $\omega$-automaton with B\"{u}chi acceptance condition.
\end{definition}


The above definition of GR(1) is the extended version of the standard General Reactivity(1) fragment introduced in~\cite{gr1def}. The advantage of using GR(1) instead of full LTL as specification language is that realizability for LTL is 2EXPTIME-complete~\cite{PR89}, whereas for GR(1) it is only cubic in the size of the formula~\cite{gr1def}. Given a finite number of deterministic $\omega$-automata with B\"uchi acceptance conditions, we can construct a deterministic $\omega$-automaton with B\"uchi acceptance condition that accepts the intersection of the languages of the given automata
~\cite{baierbook}. Thus a GR(1) formula can be converted to a deterministic $\omega$-automaton with a B\"uchi implication acceptance condition.


\subsection{Games}\label{subsec:games}
\newcommand{\cop}{\mathsf{coop}}
\newcommand{\set}[1]{\{ #1 \}}
\newcommand{\almost}{\mathsf{Almost}}

In this work, we consider the following probabilistic games that generalize Markov decision processes (MDPs). 

\begin{definition}[2\nicefrac{1}{2}-player games]
A two-player turn-based probabilistic game, or 2\nicefrac{1}{2}-player game, is a tuple $\G=(S_1,S_2,\act,\delta)$, where $S_1$ and $S_2$ are disjoint finite sets of states for Player~1 and Player~2, respectively, $\act$ is a finite set of actions for the players, and $\delta\colon (S_1 \cup S_2) \times \act\to \mathcal{D}(S_1 \cup S_2)$, is a probabilistic transition function. 
\end{definition}


Let $S=S_1 \cup S_2$. 
A play of a 2\nicefrac{1}{2}-player game $\G$ is a sequence $g\in S^{\omega}$ such that for all $n\geq 1$ there exists $a\in \act$ such that $\delta(g(n),a)(g(n+1))>0$. 
A finite play is a finite prefix of a play of $\G$. 
A Player~1 strategy for $\G$ is a function $C_{\G}^1\colon S^*\cdot S_1\to \act$ that determines the Player 1 action to be applied after any finite prefix of a play ending in a Player 1 state, and strategies for Player~2 
are defined analogously. 
If there exists an implementation of a strategy that uses finite memory, \eg a finite-state transducer, the strategy is called finite-memory. If there exists an implementation that uses only one memory element, it is called memoryless. Given a Player 1 and Player 2 strategy, and a starting state, there exists a unique probability measure over sets of plays.

Given a game $\G$, an acceptance condition defines the set of accepting plays. We consider GR(1) formulae and B\"uchi implication over $S$ as accepting conditions for $\G$. The almost-sure winning set, denoted as $\almost^{\G}(\varphi)$ for a GR(1) formula and $\almost^{\G}((E,F))$ for a B\"uchi implication condition, is the set of states such that Player~1 has a strategy to ensure the objective with probability~1 irrespective of the strategy of Player~2. Formally, $\almost^{\G}(\varphi) =\{s \in S \mid \exists C_{\G}^1 \, \forall C_{\G}^2$ the probability to satisfy $\varphi$ using the two strategies and starting from state $s$ is~1$\,\}$, and $\almost^{\G}((E,F))$ is defined similarly. The almost-sure winning set $\almost^{\G}((E,F))$ for B\"uchi implication acceptance condition can be solved in quadratic time~\cite{ChaThesis,twoandhalfplayergames}. In this work, we use more intuitive, cubic time algorithm described in detail in App.~\ref{app:gamesalg}. 
Moreover, in the states of the set $\almost^{\G}((E,F))$, there always exist witness strategies, called almost-sure winning strategies, that are memoryless and indeed pure, \ie not randomized, as defined above. This follows from the fact that the B\"{u}chi implication condition can be seen as a special case of a more general parity acceptance condition~\cite{twoandhalfplayergames}. In Sec.~\ref{sec:solution}, we show how to compute the almost-sure winning set $\almost^{\G}(\varphi)$ for a GR(1) formula $\varphi$. 
 

In this work, we also consider the following cooperative interpretation of 2\nicefrac{1}{2}-player games which is an MDP or so called 1\nicefrac{1}{2}-player game. 

\begin{definition}[1\nicefrac{1}{2}-player games]
An MDP or 1\nicefrac{1}{2}-player game, is a tuple $\G=(S,\act,\delta)$, where $S, \act$ are non-empty finite sets of states and actions, and $\delta\colon S\times \act\to \mathcal{D}(S)$ is a probabilistic transition function.
\end{definition}

Given a 2\nicefrac{1}{2}-player game $\G$, the 1\nicefrac{1}{2}-player interpretation, where the players cooperate, is called $\G^{\cop}$ with $S = S_1 \cup S_2$. The almost-sure winning set in $\G^{\cop}$ for a GR(1) formula $\varphi$ is then defined as $\almost^{\G^{\cop}}(\varphi) =\{s \in S \mid \exists C_{\G}^1\, \exists C_{\G}^2 $ such that the probability to satisfy $\varphi$ using the two strategies and starting from state $s$ is~1$\,\}$, analogously for a B\"uchi implication acceptance condition.


\section{Problem Formulation}\label{sec:pf}
In this work, we assume we are given a linear stochastic system $\T$ defined as
\begin{equation}\label{eq:linstoch}
\T: \, x_{t+1}=A  x_t + B u_t + w_t, 
\end{equation}
where $x_t\in \X \subset \mathbb{R}^N, u_t\in \U\subset \mathbb{R}^M$, $\X, \U$ are polytopes in the corresponding Euclidean spaces called the state space and control space, respectively, $w_t\in \W\subset \mathbb{R}^N$ is the value at time $t$ of a random vector with values in polytope $\W$. The random vector has positive density on all values in $\W$. Finally, $A$ and $B$ are matrices of appropriate sizes.

The system $\T$ evolves in traces. A trace of a linear stochastic system $\T$ is an infinite sequence $\rho\in \X^{\omega}$ such that for every $n\geq 1$, we have $\rho(n+1)=A\rho(n)+Bu+w$ for some $u\in \U,w\in \W$. A finite trace $\sigma\in \X^+$ of $\T$ is then a finite prefix of a trace. A linear stochastic system $\T$ can be controlled using control strategies, where a control strategy is a function $C_{\T}\colon \X^+ \to \U$.

To formulate specifications over the linear stochastic system $\T$, we assume we are given a set $\Pi$ of linear predicates over the state space $\mathcal{X}$ of $\T$:
\begin{equation}\label{eq:linpred}
\Pi=\{\pi_k\mid \pi_k:c_k^T x\leq d_k,c_k\in \mathbb{R}^N,d_k\in \mathbb{R},k\in K\},
\end{equation}
where $K$ is a finite index set. Every trace of the system generates a word over $2^\Pi$, and every GR(1) specification formulated over the alphabet $\Pi$ can be interpreted over these words.


\begin{problem}\label{pf}
Given a linear stochastic system $\T$ (Eq.~(\ref{eq:linstoch})), a finite set of linear predicates $\Pi$ (Eq.~(\ref{eq:linpred})) and a GR(1) formula $\varphi$ over alphabet $\Pi$, find the set $\X_{init}$ of states $x\in \X$ for which there exists a control strategy $C_{\T}$ such that the probability that a trace starting in state $x$ using $C_{\T}$ satisfies $\varphi$ is 1, and find the corresponding strategies for $x\in \X_{init}$.
\end{problem}

\begin{figure}[t]
\begin{center}
\scalebox{0.45}{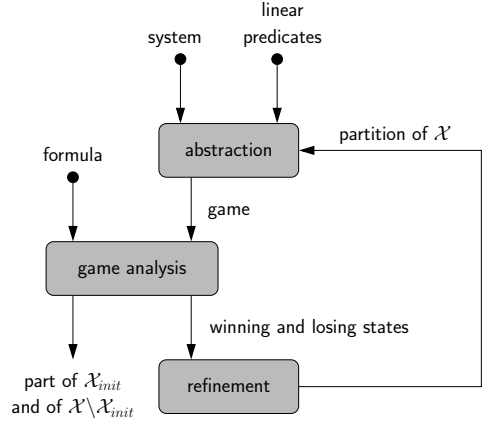}
\end{center}
\caption{Graphical representation of the proposed solution to Problem~\ref{pf}.}
\label{fig:approach}
\end{figure}

\noindent\textbf{Approach overview.} The solution we propose for Problem~\ref{pf} can be summarized as follows. First, we abstract the linear stochastic system $\T$ using a 2\nicefrac{1}{2}-player game based on the partition of the state space $\X$ given by linear predicates $\Pi$. The game is built only using polytopic operations on the state space and control space. We analyze the game and identify those partition elements of the state space $\X$ that provably belong to the solution set $\X_{init}$, as well as those that do not contain any state from $\X_{init}$. The remaining parts of the state space still have the potential to contribute to the set $\X_{init}$ but are not decided yet due to coarse abstraction. In the next step, the partition of state space $\X$ is refined using deep analysis of the constructed game. Given the new partition, we build a new game and repeat the analysis. The approach can be graphically represented as shown in Fig.~\ref{fig:approach}. 

\begin{table*}
\caption{Definitions of polytopic operators $\post$ (posterior), $\preo$ (predecessor), $\prer$ (robust predecessor), $\prep$ (precise predecessor), $\attr$ (attractor) and $\attrr$ (robust attractor), where $\X'\subseteq \X, \U'\subseteq \U$ are polytopes, and $\{\X_j\}_{j\in J}$ is a set of polytopes in $\X$. The algorithms to compute all the operators are listed in App.~\ref{app:polytopicop}. 
}\label{tab:polytopeops}
\begin{center}
{
\renewcommand{\arraystretch}{1.4}
\begin{tabular}{r c l}
\hline
$\post(\X',\U')$ & $=$ & $\{ x\in \mathbb{R}^N \mid \exists x'\in \X',\exists u\in \U',\exists w\in \W:\, x=Ax'+Bu+w\}$\\
$\preo(\X',\U',\{\X_j\}_{j\in J})$ & $=$ & $\{x\in \X'\mid \exists u\in \U':\, \post(x,u) \cap \bigcup \limits_{j\in J}\X_j \text{ is non-empty}\}$\\
$\prer(\X',\U',\{\X_j\}_{j\in J})$ & $=$ & $\{x\in \X'\mid \exists u\in \U':\, \post(x,u) \subseteq \bigcup \limits_{j\in J}\X_j\}$\\
$\prep(\X',\U',\{\X_j\}_{j\in J})$ & $=$ & $\{x\in \X'\mid \exists u\in \U':\, \post(x,u) \subseteq \bigcup \limits_{j\in J}\X_j \text{ and}$\\
& & \qquad \qquad \qquad \qquad \, \,$\forall j\in J:\, \post(x,u)\cap \X_j \text{ is non-empty}\}$\\
$\attr(\X',\U',\{\X_j\}_{j\in J})$ & $=$ & $\{x\in \X'\mid \forall u\in \U':\, \post(x,u) \cap \bigcup \limits_{j\in J}\X_j \text{ is non-empty}\}$\\
$\attrr(\X',\U',\{\X_j\}_{j\in J})$ & $=$ & $\{x\in \X'\mid \forall u\in \U':\, \post(x,u) \subseteq \bigcup \limits_{j\in J}\X_j\}$\\
\hline
\end{tabular}
}
\end{center}
\end{table*}

We prove that the result of every iteration is a partial solution to Problem~\ref{pf}. In other words, the computed set of satisfying initial states as well as the set of non-satisfying initial states are correct. Moreover, they are improved or maintained with every iteration as the abstraction gets more precise. This allows us to efficiently use the proposed algorithm for a fixed number of iterations. Finally, we prove that if the algorithm terminates then the result is indeed the solution to Problem~\ref{pf}. 

The main difficulty of the approach is the abstraction-refinement of 2\nicefrac{1}{2}-player game. Abstraction-refinement has been considered for discrete systems~\cite{norman_fmsd10,HJM03,CCD14,CHJM05}, and also for some classes of hybrid systems~\cite{norman_qest11, ozay_cdc14}. However, in all of these approaches, even if the original system is considered to be probabilistic, the distributions are assumed to be discrete and given, and are not abstracted away during the refinement. The key challenge is the extension of abstraction-refinement approach to continuous stochastic systems, where the transition probabilities in the abstract discrete model need to be abstracted. We show that by exploiting the nature of the considered dynamic systems we can develop an abstraction-refinement approach for our problem, see Rem.~\ref{remark:interplay}.



\section{Solution}\label{sec:solution}

In this section, we describe the proposed solution in detail and present necessary proofs. We start with the abstraction procedure that consists of two steps. The linear stochastic system $\T$ is first abstracted using a non-deterministic transition system which is then extended to a 2\nicefrac{1}{2}-player game. The game analysis section then describes how to identify parts of the solution to Problem~\ref{pf}. The procedure for refinement is presented last. 
Finally, we prove all properties of the proposed solution.  

Let $\X_{\out}$ be the set of all states outside of the state space $\X$ that can be reached within one step in system $\T$, \ie $\X_\out$ is the set $\post(\X,\U)\backslash \X$, where $\post$ is the posterior operator defined in Tab.~\ref{tab:polytopeops}. 
Note that $\X_{\out}$ is generally not a polytope, but it can be represented as a finite set of polytopes $\{\X_{i_{\out}}\}_{i_{\out}\in I_{\out}}$, or $\{\X_{i_{\out}}\}$ for short. 
All polytopic operators that are used in this section are formally defined in Tab.~\ref{tab:polytopeops} and their computation is described in detail in App.~\ref{app:polytopicop}.

\subsection{Abstraction}\label{subsec:abstraction}

The abstraction consists of two steps. First, the linear stochastic system is abstracted using a non-determi\-nis\-tic transition system which is then extended to a 2\nicefrac{1}{2}-player game. 

\begin{definition}[NTS]\label{def:nts}
A non-deterministic transition sys\-tem (NTS) is a tuple $\N=(S,\act,\delta)$, where $S$ is a non-empty finite set of states, $\act$ is a non-empty finite set of actions, and $\delta\colon S\times \act \to 2^S$ is a non-deterministic transition function.
\end{definition}



\medskip

\noindent\textbf{NTS construction.} In order to build an NTS abstraction for $\T$, we assume we are given a partition $\{\X_i\}_{i\in I}$, or $\{\X_i\}$ for short, of the state space $\X$. Initially, the partition is given by the set of linear predicates $\Pi$, \ie it is the partition given by the equivalence relation $\sim_\Pi$ defined as
\begin{equation*}
x \sim_{\Pi} x'\quad \Longleftrightarrow \quad \forall k\in K\colon \big( \, c_k^T x\leq d_k \, \Leftrightarrow \, c_k^T x'\leq d_k \big).
\end{equation*}
In the later iterations of the algorithm, the partition is given by the refinement procedure. The construction below builds on the approach from \cite{boyan}.

We use $\N_{\{\X_i\}}=(S_{\N}, \act_{\N}, \delta_{\N})$ to denote the NTS corresponding to partition $\{\X_i\}$ defined as follows. The states of $\N_{\{\X_i\}}$ are given by the partition of the state space $\X$ and the outer part $\X_{\out}$, \ie $S_\N=\{\X_i\}\cup \{\X_{i_{\out}}\}$.   
Let $\X_i\in \{\X_i\}\subset S_\N$ be a state of the NTS, a polytope in $\X$.  
We use $\sim_i$ to denote the equivalence relation on $\U$ such that $u \sim_i u'$ if for every state $\X_j\in \{\X_i\}\cup \{\X_{i_{\out}}\}$, it holds that $\post(\X_i,u)\, \cap\, \X_j$ is non-empty if and only if $\post(\X_i,u')\, \cap\, \X_j$ is non-empty. 
Intuitively, two control inputs are equivalent with respect to $\X_i$, if from $\X_i$ the system $\T$ can transit to the same set of partition elements of $\X$ and $\X_{\out}$. The partition $\U/\sim_i$ is then the set of all actions of the NTS $\N_{\{\X_i\}}$ that are allowed in state $\X_i$. We use $\U_i^J$ to denote the union of those partition elements from $\U/\sim_i$ that contain control inputs that lead the system from $\X_i$ to polytopes $\X_j,j\in J\subseteq I\cup I_{\out}$, \ie
\begin{align}\label{eq:UiJ}
\U_i^J = \{ u\in \U \mid & \forall j\in J:\, \post(\X_i,u)\cap \X_j \text{ is non-empty and} \nonumber \\
& \forall j\not \in J:\, \post(\X_i,u)\cap \X_j \text{ is empty} \}.
\end{align}
The set $\U_i^J$ can be computed using only polytopic computations as described in App.~\ref{subapp:UiJ}. 
For a state $\X_i\in \{\X_i\}\subset S_\N$ and action $\U_{i'}^J\in \act_\N$, we let
\begin{equation*}
\delta_\N(\X_i,\U_{i'}^J)=\begin{cases}
\{\X_j \mid j\in J\} & \text{if }i=i',\\
\emptyset & \text{otherwise.}
\end{cases}
\end{equation*}
For states $\X_{i_{\out}}\in \{\X_{i_{\out}}\}\subset S_\N$, no actions or transitions are defined.

\medskip

\noindent\textbf{From NTS to game.} Since the NTS does not capture the probabilistic aspect of the linear stochastic system, we build a 2\nicefrac{1}{2}-player game on top of the NTS. Let $\X_i$ be a polytope within the state space $\X$ of $\T$, a state of $\N_{\{\X_i\}}$. When $\T$ is in a particular state $x\in \X_i$ and a control input $u\in \U_i^J$ is to be applied, we can compute the probability distribution over the set $\{\X_j\}_{j\in J}$ that determines the probability of the next state of $\T$ being in $\X_j,j\in J$, using the distribution of the random vector for uncertainty. The evolution of the system can thus be seen as a game, where Player 1 acts in states $\X_i\in S_\N$ of the NTS and chooses actions from $\act_\N$, and Player 2 determines the exact state within the polytope $\X_i$ and thus chooses the probability distribution according to which a transition in $\T$ is made. This intuitive game construction implies that Player 2 has a possibly infinite number of actions. On the 
other hand, in Problem~\ref{pf} we are interested in satisfying the GR(1) specification with probability 1 and in the theory of finite discrete probabilistic models, it is a well-studied phenomenon that in almost-sure analysis, the exact probabilities in admissible probability distributions of the model are not relevant. It is only important to know supports of such distributions, see \eg \cite{baierbook}. That means that in our case we do not need to consider all possible probability distributions as actions for Player 2, but it is enough to consider that Player 2 chooses support for the probability distribution that will be used to make a transition. For a polytope $\X_i\in S_\N$ and $\U_i^J\in \act_\N$, we use $\supp(\X_i,\U_i^J)$ to denote the set of all subsets $J'\subseteq J$ for which there exist $x\in \X_i,u\in \U_i^J$ such that the next state $x'=Ax+Bu+w$ of $\T$ belongs to $\X_j,j\in J'$ with non-zero probability and with zero probability to $\X_j,j\not \in J'$, \ie
\begin{align}\label{eq:supp}
\supp(\X_i,\U_i^J) = \{ J'\subseteq J \mid & \prep(\X_i,\U_i^J,\{\X_j\}_{j\in J'}) \nonumber \\
& \text{is non-empty}\},
\end{align} 
where $\prep$ is the precise predecessor operator from Tab.~\ref{tab:polytopeops}. 

\medskip

\noindent\textbf{Game construction.} Given the NTS $\N_{\{\X_i\}}$, the $2\nicefrac{1}{2}$-player game $\G_{\{\X_i\}}=(S_1,S_2,\act,\delta)$ is defined as follows. Player 1 states $S_1=\{\X_i\}\cup \{\X_{i_{\out}}\}$ are the states $S_\N$ of the NTS and Player 1 actions are the actions $\act_\N$ of $\N_{\{\X_i\}}$. Player 2 states are given by the choice of an action in a Player 1 state, \ie $S_2= \{\X_i\}\times \{\U_i^J\}$. The Player 2 actions available in a state $(\X_i,\U_i^J)$ are the elements of the set $\supp(\X_i,\U_i^J)$ defined in Eq.~(\ref{eq:supp}). For Player 1, the transition probability function $\delta$ defines non-zero probability transitions only for triples of the form $\X_i,\U_i^J, (\X_i,\U_i^J)$ and for such it holds $\delta(\X_i,\U_i^J)((\X_i,\U_i^J))=1$. For Player 2, the function $\delta$ defines the following transitions:
\begin{equation*}
\delta_2 \big((\X_i,\U_i^J),J'\big)\big(\X_j\big) = \begin{cases}
\frac{1}{|J'|} & \text{if }J'\in \supp(\X_i,\U_i^J)\\
& \text{and } j\in J',\\
0 & \text{otherwise}.
\end{cases}
\end{equation*}
The definition reflects the fact that once Player 2 chooses the support, the exact transition probabilities are irrelevant and without loss of generality, we can consider them to be uniform.

\begin{example}{\scshape{(Illustrative example, part I)}}\label{ex:ill1}
Let $\T$ be a linear sto\-chastic system of the form given in Eq.~(\ref{eq:linstoch}), where
\begin{equation*}
A=\begin{pmatrix}
1 & 0\\0 & 1
\end{pmatrix}, B = \begin{pmatrix}
1 & 0\\0 & 1
\end{pmatrix},
\end{equation*}
the state space is 
$\X=\{x\in \mathbb{R}^2\mid 0\leq x(1)\leq 4, 0\leq x(2)\leq 2\},$ 
the control space is  
$\U=\{u\in \mathbb{R}^2\mid -1\leq u(1),u(2)\leq 1\},$ 
and the random vector takes values in polytope $\W=\{w\in \mathbb{R}^2\mid -0.1\leq w(1),w(2)\leq 0.1\}$. Let $\Pi$ contain a single linear predicate $\pi_1\colon\, x(1)\leq 2$. In Fig.~\ref{fig:exgame1}, polytopes $\X_1$ and $\X_2$ form the partition of $\X$ given by $\Pi$, and polytopes $\X_3,\X_4,\X_5,\X_6$ form the rest of the one step reachable set of system $\T$, \ie $\X_\out$. The game $\G_{\{\X_i\}}$ given by this partition has 6 states and 18 actions. In Fig.~\ref{fig:exgame2}, we visualize part of the transition function as follows. In Player 1 state $\X_1$, if Player 1 chooses, \eg action $\U_1^{\{1,2,5\}}$ that leads from $\X_1$ to polytopes $\X_1,\X_2,\X_5$, the game is in Player 2 state $(\X_1,\U_1^{\{1,2,5\}})$ with probability 1. Actions of Player 2 are the available supports of the action over the set $\{\X_1,\X_2,\X_5\}$, which are in this case all non-empty subsets. If Player 2 chooses, \eg support $\{\X_1,\X_2\}$, the game is in Player 1 state $\X_1$ or $\X_2$ with equal probability $0.5$.

\begin{figure}[t]
\begin{center}
\includegraphics[scale=0.25]{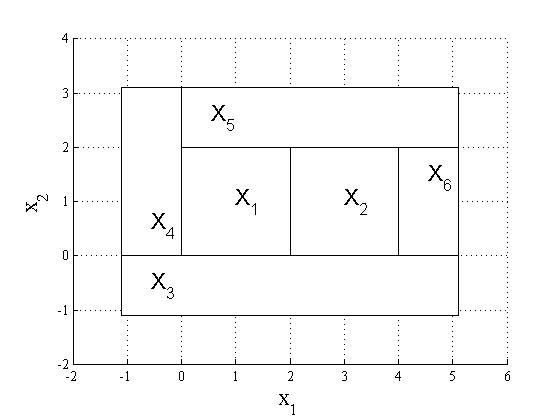} 
\end{center}
\caption{Partition of state space $\X$ of system $\T$ in Ex.~\ref{ex:ill1} given by linear predicates $\Pi$. Polytopes $\X_3,\ldots,\X_6$ form the set $\X_\out$.}\label{fig:exgame1}
\end{figure}

\begin{figure}[t]
\begin{center}
\scalebox{0.3}{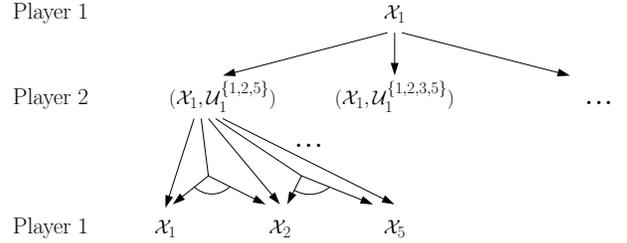} 
\end{center}
\caption{Part of the transition function of the game $\G_{\{\X_i\}}$ constructed in Ex.~\ref{ex:ill1}.}\label{fig:exgame2}
\end{figure}
\end{example}

The following proposition proves that the game $\G_{\{\X_i\}}$ simulates the linear stochastic system $\T$. 

\begin{proposition}
Let $\rho$ be a trace of the linear stochastic system $\T$. 
Then there exists a play $g$ of the game $\G_{\{\X_i\}}$ such that $\rho(n)\in g(2n-1)$ for every $n\geq 1$.
\end{proposition}

\begin{proof}
The play $g$ is defined as follows. The states $g(2n-1)=\X_i$ such that $\rho(n)\in \X_i$. The states $g(2n)=(\X_i,\U_i^J)$ such that there exist $u\in \U_i^J,w\in \W$ for which $\rho(n+1)=A\rho(n)+Bu+w$.  
\end{proof}

On the other hand, since $\G_{\{\X_i\}}$ is only an abstraction of the system $\T$, it may contain plays that do not correspond to any trace of the system. 

\subsection{Game analysis}\label{subsec:gameanalysis}
\newcommand{\yes}{\mathit{yes}}
\newcommand{\no}{\mathit{no}}

Let $\G_{\{\X_i\}}$ be the 2\nicefrac{1}{2}-player game constructed for the linear stochastic system $\T$ and partition $\{\X_i\}$ of its state space using the procedure from Sec.~\ref{subsec:abstraction}. In this section, we identify partition elements from $\{\X_i\}$ which are part of the solution set of initial states $\X_{init}$ as well as those that do not contain any satisfying initial states from $\X$. 

\medskip

\noindent\textbf{Computing satisfying states.} First, we compute the almost-sure winning set $S_{\yes}$ in game $\G_{\{\X_i\}}$ with respect to the GR(1) formula $\varphi$ from Problem~\ref{pf}, \ie
\begin{equation}\label{eq:syes}
S_{\yes} = \almost^{\G_{\{\X_i\}}}(\varphi).
\end{equation}
We proceed as follows. Let $\mathcal{A}=(Q,2^{\Pi},\delta^{\mathcal{A}},q_0,(E,F))$ be a deterministic $\omega$-automaton  with B\"{u}chi implication acceptance condition for the GR(1) formula $\varphi$ constructed as described in Sec.~\ref{subsec:gr1}. We consider the 2\nicefrac{1}{2}-player game $\mathcal{P}_{\{\X_i\}}=(S_1^{\mathcal{P}},S_2^{\mathcal{P}},\act,\delta^{\mathcal{P}})$ that is the synchronous product of $\G_{\{\X_i\}}$ and $\mathcal{A}$, \ie $S_1^{\mathcal{P}}=S_1\times Q$, $S_2^{\mathcal{P}}=S_2\times Q$, and for every $(\X_i,q)\in S_1^{\mathcal{P}}$ and $\U_i^J\in \act$ we have
\begin{equation*}\small
\delta^{\mathcal{P}}\big((\X_i,q),\U_i^J\big)\big((\X_i,\U_i^J),q')\big) =
\begin{cases}
\delta \big(\X_i,\U_i^J\big)\big((\X_i,\U_i^J)\big) & \\
\quad \text{if }\delta^{\mathcal{A}}(q,\Pi(\X_i))=q', & \\
0 & \\
\quad \text{otherwise}, &
\end{cases}
\end{equation*}
where $\Pi(\X_i)$ is the set of all linear predicates from $\Pi$ that are true on polytope $\X_i$, 
and similarly, for all $((\X_i,\U_i^J),q) \in S_2^{\mathcal{P}}$ and $J'\in \act$ we have
\begin{equation*}\small
\delta^{\mathcal{P}}\big(((\X_i,\U_i^J),q),J'\big)\big((\X_j,q')\big) =
\begin{cases}
\delta \big((\X_i,\U_i^J),J'\big)\big(\X_j\big) & \\
\quad \text{if }q=q', & \\
0 & \\
\quad \text{otherwise}. &
\end{cases}
\end{equation*}
When constructing the product game, we only consider those states from $S_1\times Q$ and $S_2\times Q$ that are reachable from some $(\X_i,q_0)$, where $q_0$ is the initial state of the automaton $\mathcal{A}$. Finally, we consider B\"{u}chi implication acceptance condition $(E^{\mathcal{P}},F^{\mathcal{P}})$, where $E^{\mathcal{P}}=(S_1^{\mathcal{P}}\cup S_2^{\mathcal{P}})\times E$ and $F^{\mathcal{P}}=(S_1^{\mathcal{P}}\cup S_2^{\mathcal{P}})\times F$.

\begin{proposition}\label{prop:syesproduct}
The set $S_{\yes}$ defined in Eq.~(\ref{eq:syes}) consists of all $\X_i\in S_1$ for which $(\X_i,q_0)\in S_\yes^{\mathcal{P}}$, where the set 
\begin{equation}\label{eq:syesp}
S_{\yes}^{\mathcal{P}} = \almost^{\mathcal{P}_{\{\X_i\}}}((E^{\mathcal{P}},F^{\mathcal{P}}))
\end{equation}
can be computed using algorithm in App.~\ref{app:gamesalg}.
\end{proposition}

\begin{proof}
Follows directly from the construction of the game $\mathcal{P}_{\{\X_i\}}$ above and the results of~\cite{CdAH11}.
\end{proof}

The next proposition proves that the polytopes from $S_{\yes}$ are part of the solution to Problem~\ref{pf}.

\begin{proposition}\label{prop:syes}
For every $\X_i\in S_{\yes}$, there exists a finite-memory strategy $C_{\T}$ for $\T$ such that every trace of $\T$ under strategy $C_{\T}$ that starts in any $x\in \X_i$ satisfies $\varphi$ with probability 1.
\end{proposition}

\begin{proof}
Let $\X_i\in S_{\yes}$ and let $C_{\G_{\{\X_i\}}}$ be a finite-memory almost-sure winning strategy for Player 1 from state $\X_i$ in game $\G_{\{\X_i\}}$, see Sec.~\ref{subsec:games}. Let $C_{\T}$ be a strategy for $\T$ defined as follows. For a finite trace $\sigma_{\T}$, let $C_{\T}(\sigma_{\T})=u$, where $u\in C_{\G_{\{\X_i\}}}(\sigma_{\G_{\{\X_i\}}})$, where $\sigma_{\G_{\{\X_i\}}}$ is finite play such that $\sigma_{\T}(n)\in \sigma_{\G_{\{\X_i\}}}(2n)$ for every $1\leq n\leq |\sigma_{\T}|$. Since $C_{\G_{\{\X_i\}}}$ for game $\G_{\{\X_i\}}$ is almost-sure winning from state $\X_i$ with respect to $\varphi$, \ie every play that starts in $\X_i$ almost-surely satisfies $\varphi$, the analogous property holds for $C_{\T}$ and traces in $\T$.
\end{proof}

\medskip

\noindent \textbf{Computing non-satisfying states.} Next, we consider the set $S_{\no}$ of Player 1 states in game $\G_{\{\X_i\}}$ defined as follows:
\begin{equation}\label{eq:sno}
S_{\no} = S_1 \setminus  \almost^{\G^{\cop}_{\{\X_i\}}}(\varphi).
\end{equation}
Intuitively, $S_{\no}$ is the set of states, where even if Player~2 cooperates with Player~1, $\varphi$ can still not be satisfied with probability~1. 

\begin{proposition}\label{prop:snoproduct}
The set $S_{\no}$ defined in Eq.~(\ref{eq:sno}) consists of all $\X_i\in S_1$ for which $(\X_i,q_0)\in S_\no^{\mathcal{P}}$, where 
\begin{equation}\label{eq:snop}
S_{\no}^{\mathcal{P}} = S_1^{\mathcal{P}} \setminus  \almost^{\mathcal{P}^{\cop}_{\{\X_i\}}}((E^{\mathcal{P}},F^{\mathcal{P}})).
\end{equation}
\end{proposition}

\begin{proof}
Follows directly from the construction of the product game $\mathcal{P}_{\{\X_i\}}$.
\end{proof}

We prove that no state $x\in \X_i$ for $\X_i\in S_\no$ is part of the solution to Problem~\ref{pf}.

\begin{proposition}\label{prop:sno}
For every $\X_i\in S_{no}$ and $x\in \X_i$, there does not exist a strategy $C_{\T}$ for $\T$ 
such that every trace of $\T$ under $C_{\T}$ starting in $x$ satisfies $\varphi$ with probability 1.
\end{proposition}

\begin{proof}
Intuitively, from the construction of the game $\G_{\{\X_i\}}$ in Sec.~\ref{subsec:abstraction}, Player 2 represents the 
unknown precise state of the system $\T$ within in the abstraction, \ie he makes the choice of a state inside each polytope $\X_i$ at each step.
Therefore, if $\varphi$ cannot be almost-surely satisfied from $\X_i$ in the game even if the two players cooperate, in $\T$ it translates to the fact that $\varphi$ cannot be almost-surely satisfied from any $x\in \X_i$ even if we consider 
strategies that can moreover change inside each $\X_i$ arbitrarily at any moment.
\end{proof}

\medskip

\noindent \textbf{Undecided states.} Finally, consider the set 
\begin{equation}\label{eq:s?}
S_{?}=S_1\backslash (S_{\yes}\cup S_{\no}).
\end{equation}
These are the polytopes within the state space of $\T$ that have not been decided as satisfying or non-satisfying due to coarse abstraction. 
Alternatively, from Prop.~\ref{prop:syesproduct} and \ref{prop:snoproduct}, and Eq.~(\ref{eq:s?}), we can define the set $S_?$ as the set of all $\X_i\in S_1$, for which $(\X_i,q_0)\in S_?^\product$, where
\begin{equation}\label{eq:s?p}
S_{?}^\product=S_1^\product \backslash (S_{\yes}^\product\cup S_{\no}^\product).
\end{equation} 

\begin{proposition}\label{prop:s?}
For every $\X_i\in S_?$ it holds that the product game $\product_{\{\X_i\}}$ can be won cooperatively starting from the Player 1 state $(\X_i,q_0)$. Analogously, for every $(\X_i,q)\in S_?^\product$ it holds that the product game $\product_{\{\X_i\}}$ can be won cooperatively starting from $(\X_i,q)$.
\end{proposition}

\begin{proof}
The proposition follows directly from Eq.~(\ref{eq:s?}) and (\ref{eq:s?p}), and Prop.~\ref{prop:syesproduct} and~\ref{prop:snoproduct}.
\end{proof}

\begin{example}{\scshape{(Illustrative example, part II)}}\label{ex:ill2}
Recall the linear stochastic system $\T$ from Ex.~\ref{ex:ill1} and consider GR(1) formula $\mathbf{F}\, \neg \pi_1$ over the set $\Pi$ that requires to eventually reach a state $x\in \X$ such that $x(1)\geq 2$. The deterministic $\omega$-automaton for the formula has only two states, $q_0$ and $q_1$. The automaton remains in the initial state $q_0$ until polytope $\X_2$ is visited in $\T$. Then it transits to state $q_1$ and remains there forever. The B\"{u}chi implication condition $(E,F)$ is $E=\{q_0\},F=\{q_1\}$. The solution of the game $\G_{\{\X_i\}}$ constructed in Ex.~\ref{ex:ill1} with respect to the above formula is depicted in Fig.~\ref{fig:exgame3}. 

\begin{figure}[t]
\begin{center}
\includegraphics[scale=0.25]{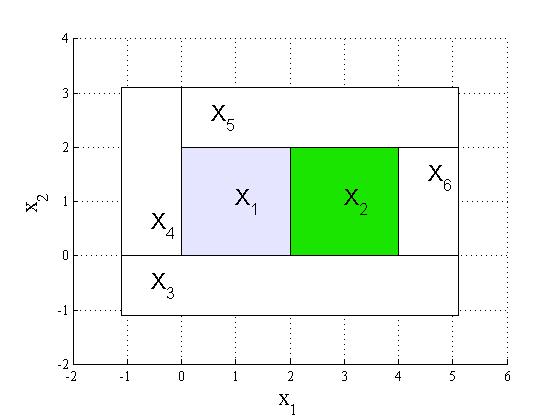} 
\end{center}
\caption{Solution of the game in Ex.~\ref{ex:ill2}. The polytopes, \ie Player 1 states, that belong to sets $S_{yes}, S_{no},S_?$ are shown in green, white and light blue, respectively.}\label{fig:exgame3}
\end{figure}
\end{example}

If the set $S_\no$ contains all Player 1 states of the game $\G_{\{X_i\}}$, the GR(1) formula $\varphi$ cannot be satisfied in the system $\T$ and our algorithm terminates. If set $S_?$ is empty, the algorithm terminates and returns the union of all polytopes from $S_\yes$ as the solution to Problem~\ref{pf}. The corresponding satisfying strategies are synthesized as described in the proof of Prop.~\ref{prop:syes}. Otherwise, we continue the algorithm by computing a refined partition of the state space $\X$ as described in the next section. 

\subsection{Refinement}\label{subsec:refinement}
Refinement is a heuristic that constructs a new partition of $\X$, a subpartition of $\{\X_i\}$, that is used in the next iteration of the overall algorithm. We design two refinement procedures, called positive and negative, that aim to enlarge the combined volume of polytopes in the set $S_{yes}$ and $S_{no}$, respectively, or equivalently, to reduce the combined volume of polytopes in the set $S_?$. Based on Prop.~\ref{prop:syesproduct} and \ref{prop:snoproduct}, both procedures are formulated over the product game $\product_{\{\X_i\}}$ and reach their respective goals through refining polytopes $\X_i$ for which $(\X_i,q)\in S_?^\product$ for some $q\in Q$.

In this section, we use $J_{yes}^{q}$ to denote the set of all indices $i\in I$ for which $(\X_i,q)\in S_{yes}^\product$, and $J_?^{q},J_{no}^{q}$ are defined analogously. In the two refinement procedures, every polytope $\X_i$ can be partitioned into a set of polytopes in iterative manner, as $(\X_i,q)\in S_?^\product$ can hold for multiple $q\in Q$. Therefore, given a partition of $\X_i$, the refinement of $\X_i$ according to a polytope $\mathcal{B}$ refers to the partition of $\X_i$ that contains all intersections and differences of elements of the original partition of $\X_i$ and polytope $\mathcal{B}$. 

\medskip

\noindent\textbf{Positive refinement.} In the positive refinement, we explore the following property of states in $S_?^\product$. 
In Prop.~\ref{prop:s?}, we proved that the product game $\product_{\{\X_i\}}$ can be won cooperatively from every $(\X_i,q)\in S_?^\product$. It follows that there exists a Player 1 action $\U_i^J$ and Player 2 action $J'$ such that after their application in $(\X_i,q)$, the game is not in a losing state with probability 1. We can graphically represent this property as follows:
\begin{equation}\label{eq:twosteps}
(\X_i,q)\xrightarrow{\U_i^J} ((\X_i,\U_i^J),q')\xrightarrow{J'} \begin{cases}
(\X_{j_1},q') & \\
\quad \vdots & \\
(\X_{j_n},q') &
\end{cases}
\end{equation}
where an arrow $a\xrightarrow{b}$ represents the uniform probability distribution $\delta_\product(a,b)$, and $\{j_1,\ldots ,j_n\}=J'\subseteq J_{yes}^{q'}\cup J_?^{q'}$. Note that from the construction of the product game $\product_{\{\X_i\}}$ in Sec.~\ref{subsec:games} it follows that $q'$ is given uniquely over all actions $\U_i^J$. The following design ensures that every polytope $\X_i$ is refined at least once for every its appearance $(\X_i,q)\in S_?^\product, q\in Q$.

Let $(\X_i,q)\in S_?^\product$. The positive refinement first refines $\X_i$ according to the robust predecessor
\begin{equation}\label{eq:posrefprer}
\prer(\X_i,\U,\{\X_j\}_{j\in J_{yes}^{q'}}).
\end{equation}
That means, we find all states $x\in \X_i$ for which there exists any control input under which the system $\T$ evolves from $x$ to a state $x'\in \X_j, j\in J_{yes}^{q'}$.

Next, the positive refinement considers three cases. First, assume that from $(\X_i,q)$, the two players can cooperatively reach a winning state of the product game in two steps with probability 1, and let $\U_i^J$ and $J'$ be Player 1 and Player 2 actions, respectively, that accomplish that, \ie in Eq.~(\ref{eq:twosteps}), $\{j_1,\ldots ,j_n\}=J'\subseteq J_{yes}^{q'}$. For every such $\U_i^J, J'$, we find an (arbitrary) partition $\{\U_y\}_{y\in Y}$ of the polytope $\U_i^J$ 
and we partition $\X_i$ according to the robust attractors 
\begin{equation}\label{eq:posrefattrr}
\attrr(\X_i,\U_y,\{\X_j\}_{j\in J_{yes}^{q'}}).
\end{equation}
Intuitively, the above set contains all $x\in \X_i$ such that under every control input $u\in \U_y$, $\T$ evolves from $x$ to a state $x'\in \X_j,j\in J_{yes}^{q'}$. Note that the robust attractor sets partition the robust predecessor set from Eq.~(\ref{eq:posrefprer}), as every state $x$ that belongs to one of the robust attractor set must lie in the robust predecessor set as well. In the next iteration of the overall algorithm, the partition elements given by the robust attractor sets will belong to the set $S_{yes}^\product$. In the second case, assume that the two players can reach a winning state of the product game cooperatively in two steps, but only with probability $0<p<1$, while the probability of reaching a losing state is $0$. Let $\U_i^J,J'$ be Player 1 and Player 2 actions, respectively, that maximize $p$, \ie in Eq.~(\ref{eq:twosteps}), there exists $m<n$ such that $\{j_1,\ldots ,j_m\}= J'\cap J_{yes}^{q'}$, $\{j_{m+1},\ldots ,j_n\}= J'\cap J_{?}^{q'}$ and $p=\frac{m}{n}$ is maximal. 
Similarly as in the first case, we refine the polytope $\X_i$ according to the robust attractor sets as in Eq.~(\ref{eq:posrefattrr}), but we compute the sets with respect to the set of indices $J_{yes}^{q'}\cup \{j_{m+1},\ldots ,j_n\}$. 
Finally, assume that $(\X_i,q)$ does not belong to any of the above two categories.
As argued at the beginning of this section, there still exist Player 1 and Player 2 actions $\U_i^J$ and $J'$, respectively, such that in Eq.~(\ref{eq:twosteps}), 
$\{j_1,\ldots ,j_n\}=J'\subseteq J_{?}^{q'}$. Again, we refine the polytope $\X_i$ according to the robust attractor sets as in Eq.~(\ref{eq:posrefattrr}), where the sets are computed with respect to the set of indices $J_{?}^{q'}$.

\begin{example}{\scshape{(Illustrative example, part III)}}\label{ex:ill3}
We de\-monstrate a part of the the positive refinement for the game in Ex.~\ref{ex:ill2}. Consider polytope $\X_1\in S_?$. It follows from the form of the $\omega$-automaton in Ex.~\ref{ex:ill2} that $\X_1$ appears in $S_?^\product$ only in pair with $q_0$, \ie $(\X_1,q_0)\in S_?^\product$. Note that for state $(\X_1,q_0)$, every successor state is of the form $((\X_1,\U_1^J),q_0)$, \ie $q'=q_0$. First, polytope $\X_1$ is refined with respect to the robust predecessor 
\begin{equation*}
\prer(\X_1,\U,\{\X_2\}),
\end{equation*}
since $J_{yes}^{q_0}=\{\X_2\}$ because $(\X_2,q_0)\in S_{yes}^\product$ is a winning state of the product game. The robust predecessor set is depicted in Fig.~\ref{fig:exgame4} in cyan. Next, we decide which of the three cases described in the positive refinement procedure above applies to state $(\X_1,q_0)$. Consider for example Player 1 action $\U_1^{\{1,2,5\}}$ and Player 2 action $\{2\}$, as shown in Fig.~\ref{fig:exgame2}. It holds that
\begin{equation*}
(\X_1,q_0)\xrightarrow{\U_1^{\{1,2,5\}}} ((\X_1,\U_1^{\{1,2,5\}}),q_0)\xrightarrow{\{2\}} (\X_2,q_0),
\end{equation*} 
and $(\X_2,q_0)\in S_{yes}^\product$ is a winning state of the product game. Therefore, the state $(\X_1,q_0)$ is of the first type. To further refine polytope $\X_1$, we first partition the polytope
\begin{equation*}
\U_1^{\{1,2,5\}} = \{u\in \U\mid 0.1\leq u(1),u(2)\leq 1\},
\end{equation*}
\eg into 4 parts as shown in Fig.~\ref{fig:exgame4} on the right. The robust attractor 
\begin{equation*}
\attrr(\X_1,\U_3,\{\X_2\})
\end{equation*}
for one of the polytopes $\U_3$ is depicted in magenta in Fig.~\ref{fig:exgame4}. This polytope will be recognized as a satisfying initial polytope in the next iteration, since starting in any $x$ within the robust predecessor, system $\T$ as defined in Ex.~\ref{ex:ill1} evolves from $x$ under every control input from $\U_3$ to polytope $\X_2$.

\begin{figure}[t]
\begin{center}
\includegraphics[scale=0.25]{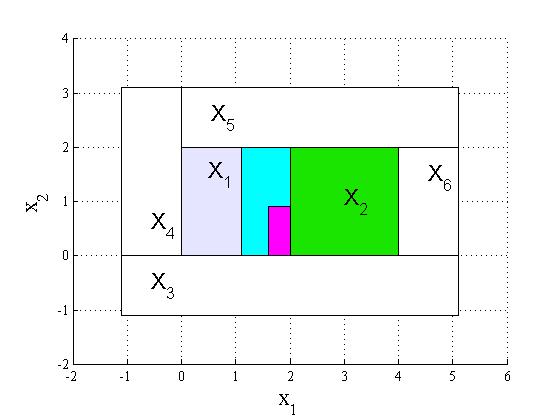} \quad \includegraphics[scale=0.15]{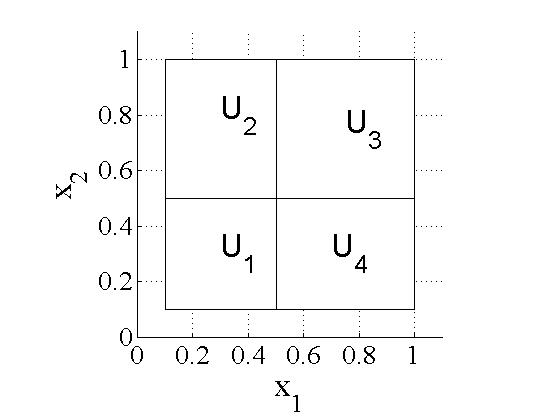}
\end{center}
\caption{Part of the positive refinement for the system in Ex.~\ref{ex:ill3}. Polytope $\X_1$ is first refined according to the robust predecessor as in Eq.~(\ref{eq:posrefprer}), the robust predecessor is shown in cyan. Next, we consider the polytope of control inputs $\U_1^{\{1,2,5\}}$ and its partition as depicted on the right. The robust predecessor of $\U_3$ is then shown in magenta.}\label{fig:exgame4}
\end{figure}
\end{example}


\noindent\textbf{Negative refinement.} In the negative refinement, we consider all Player 1 states $(\X_i,q)\in S_{?}^\product$ such that if Player 2 does not cooperate, but rather plays against Player 1, the game is lost with non-zero probability. In other words, for every Player 1 $\U_i^J$, there exists a Player 2 action $J'$ such that in Eq.~(\ref{eq:twosteps}), there exists an index $j\in J'$ such that $(\X_j,q')\in S_{no}^\product$. In this case, we refine polytope $\X_i$ according to the attractor set 
\begin{equation*}
\attr(\X_i,\U,\{\X_j\}_{j\in J_{no}^{q'}}).
\end{equation*}
Intuitively, the attractor set contains all states $x\in \X_i$ such that by applying any control input $u\in \U$, system $\T$ evolves from $x$ to a state in $\X_j$ for some $j\in J_{no}^{q'}$ with non-zero probability. In the next iteration of the algorithm, the partition elements given by the attractor set will belong to the set $S_{no}^\product$.

\begin{remark}\label{remark:interplay}
We remark that both game theoretic aspects as well as the linear stochastic dynamics play an important role in the refinement step. The game theoretic results compute the undecided states, and thereby determine what parts of the state space need to be refined and which actions need to be considered in the refinement. The linear stochastic dynamics allow us to perform the refinement itself using polytopic operators. 
\end{remark}

\subsection{Correctness and complexity}\label{subsec:solutionproperties}

We prove that the algorithm presented in this section provides a partially correct solution to Problem~\ref{pf}. 

For $n\in \mathbb{N}$, let $S_{\yes}^n,S_\no^n$ be the sets from Eq.~(\ref{eq:syes}) and (\ref{eq:sno}), respectively, computed in the $n$-th iteration of the algorithm presented above. We use $\X_\yes^n,\X_\no^n\subseteq \X$ to denote the union of polytopes from $S_{\yes}^n$ and $S_\no^n$, respectively. 

\begin{theorem}{\scshape{(Progress)}}\label{prop:algimprovement}
For every $n\in \mathbb{N}$, it holds that $\X_\yes^n\subseteq \X_\yes^{n+1}$ and $\X_\no^n\subseteq \X_\no^{n+1}$.
\end{theorem}

\begin{proof}
Follows from Prop.~\ref{prop:syesproduct} and \ref{prop:snoproduct}, and the fact that the partition of the state space $\X$ used in $n+1$-th iteration is a subpartition of the one used in $n$-th iteration.
\end{proof}

\begin{theorem}{\scshape{(Soundness)}}\label{prop:algsoundness}
For every $n\in \mathbb{N}$, it holds that $\X_\yes^n\subseteq \X_{init}$ and $\X_\no^n\subseteq \X\backslash \X_{init}$.
\end{theorem}

\begin{proof}
Follows directly from Prop.~\ref{prop:syes} and \ref{prop:sno}.
\end{proof}

\begin{theorem}{\scshape{(Partial Correctness)}}\label{prop:algpartialcorr}
If the algorithm from Sec.~\ref{sec:solution} terminates, after $n$-th iteration, then $\X_{init}=\X_\yes^n$ is the solution of Problem~\ref{pf} and the corresponding winning strategies for every $x\in \X_{init}$ are given by the winning strategies in the 2\nicefrac{1}{2}-player game from the last iteration. 
\end{theorem}

\begin{proof}
Follows directly from the condition of the algorithm termination and from Th.~\ref{prop:algimprovement} and \ref{prop:algsoundness}.
\end{proof}

It is important to note that if instead of a 2\nicefrac{1}{2}-player game a weaker abstraction model such as a 2 player game, \ie the NTS from Sec.~\ref{subsec:abstraction}, was used, our approach would not be sound. Namely, some states of $\X$ might be wrongfully identified as non-satisfying initial states based on behavior that has zero probability in the original stochastic system. In such a case, even after termination, the resulting set would only be a subset of $\X_{init}$. Therefore, the approach with 2-player games is not complete. The 2\nicefrac{1}{2}-player game is needed to account for both the non-determinism introduced by the abstraction and for the stochasticity of the system to be able to recognize (non-satisfying) behavior of zero probability.

Note that there exist linear sto\-chas\-tic systems for which our algorithm does not terminate, \ie there does not exist a finite partition of the systems' state space over which Problem~\ref{pf} can be solved for a given GR(1) formula. 

\begin{example}{\scshape{(Non-termination)}}\label{ex:infinite} 
Let $\T$ be a linear sto\-chastic system of the form given in Eq.~(\ref{eq:linstoch}), where
\begin{equation*}
A=\begin{pmatrix}
1 & 0\\0 & 1
\end{pmatrix}, B = \begin{pmatrix}
1 & 0\\0 & 1
\end{pmatrix},
\end{equation*}
state space $\X\rev{}{=\U}=\{x\in \mathbb{R}^2\mid 0\leq x(1),x(2)\leq 3\}$, control space $\U=\{u\in \mathbb{R}^2\mid -1.5\leq u(1),u(2)\leq 1.5\}$ and the random vector takes values in polytope $\W=\{w\in \mathbb{R}^2\mid -0.5\leq w(1),w(2)\leq 0.5\}$. Let $\Pi$ contain four linear predicates that partition the state space 
into a grid of three by three equally sized square polytopes. 
Assume that the aim is to eventually reach the polytope $\X_f$, where $1\leq x(1),x(2)\leq 2$. In this case, the maximal set $\X_{init}$ of states from which $\X_f$ can be reached with probability 1 is the whole state space $\X$, as for any $x\in \X$, there exists exactly one control input $u=(1.5,1.5)-x\in \U$ that leads the system $\T$ from $x$ to a state in $\X_5$ with probability 1. Since the control input is different for every $x\in \X$, there does not exist any finite state space partition, which could be used to solve Problem~\ref{pf}. 
\end{example}


\medskip

\noindent\textbf{Complexity analysis.} Finally, let us analyze the computational complexity of the designed algorithm. In the abstraction part, the 2\nicefrac{1}{2}-player game $\G_{\{\X_i\}}$ requires to first compute the set of actions for every state $\X_i, i\in I$, in time in $\mathcal{O}(2^{|I|})$ using algorithm in App.~\ref{subapp:UiJ}. 
For every action $\U_i^J$, the set of valid supports $J'\subseteq J$ is then computed in time in $\mathcal{O}(2^J)$, see App.~\ref{app:polytopicop}. 
Overall, the abstraction runs in time in $\mathcal{O}(2^{2\cdot |I|})$. The game is then analyzed using the algorithm described in App.~\ref{app:gamesalg} 
in time in $\mathcal{O}(|I|^3)$. Finally, the refinement procedure iteratively refines every polytope $\X_i$ at most $|Q|\times |\{\U_i^J\}|$ times, where $\{\U_i^J\}$ denotes the set of all actions of $\X_i$. For every $q\in Q$ such that $(\X_i,q)\in S_?^\product$, $\X_i$ is first refined using the robust predecessor operator in time exponential in $|J_{yes}^{q'}|$. 
Then $\X_i$ is refined $|Y|$ times using the robust attractor operator in polynomial time. Negative refinement is performed again for every $q\in Q$ such that $(\X_i,q)\in S_?^\product$, using the attractor operator in polynomial time. Overall, the refinement runs in time in $\mathcal{O}(|Q|\cdot 2^{|I|})$.

As the game construction is the most expensive part of the overall algorithm, the refinement procedure is designed in a way that extends both sets $S_{yes},S_{no}$ as much as possible and thus speed up convergence and minimize the number of iterations of the overall algorithm. 


\section{Case study}\label{sec:casestudy}
\begin{algorithm}[t]
\caption{Computing the set $\X_{init}\subseteq \X$ of states from which a set of polytopes $\{\X_j\}_{j\in J}$ in $\X$ can be reached with probability 1. 
}\label{alg:polytopic}
\begin{algorithmic}
\State \textbf{Input:} linear stochastic system $\T$, polytopes $\{\X_j\}_{j\in J}$
\State $\X_{>0} \gets \{\X_j\}_{j\in J}$;
\While{$\X_{>0}$ is not a fixed point}
  \State $\X_{>0}\gets \preo(\X,\U,\X_{>0})$;
\EndWhile
\State $\X_{=0,attr} \gets \X_\out \cup \X\backslash \X_{>0}$;
\While{$\X_{=0,attr}$ is not a fixed point}
  \State $\X_{=0,attr}\gets \attr(\X,\U,\X_{=0,attr});$
\EndWhile
\State $\X_{=1}\gets \X\backslash \X_{=0,attr}$;
\State \textbf{return:} $\X_{=1}$
\end{algorithmic}
\end{algorithm}

We demonstrate the designed framework on a discrete-time double integrator dynamics with uncertainties. 
Let $\T$ be a linear stochastic system of the form given in Eq.~(\ref{eq:linstoch}), where 
\begin{equation}
A =\begin{pmatrix} 1 & 1 \\ 0 & 1\end{pmatrix},\quad B=\begin{pmatrix}0.5\\1\end{pmatrix}.
\end{equation}
The state space is $\X=\{x\in \mathbb{R}^2\mid -5\leq x(1)\leq 5,-3\leq x(2)\leq 3\}$ and the control space is $\U=\{u\in \mathbb{R}\mid -1\leq u\leq 1\}$. The random vector, or uncertainty, takes values within polytope $\W=\{w\in \mathbb{R}^2\mid -0.1\leq w(1),w(2)\leq 0.1\}$. The set $\Pi$ consists of 4 linear predicates $\pi_1:\, x(1)\leq -1, \pi_2:\,x(1)\leq 1, \pi_3:\,x(2)\leq -1, \pi_4:\,x(2)\leq 1$. 
We consider GR(1) formula 
\begin{equation*}
\F (\neg \pi_1\, \wedge \, \pi_2\, \wedge \, \neg \pi_3\, \wedge \, \pi_4)
\end{equation*}
that requires the system to eventually reach a state, where both variables of the system have values in interval $(-1,1)$.

As we consider a reachability property, we can compare our approach to the algorithm shown in Alg.~\ref{alg:polytopic} that is an extension of the reachability algorithm for Markov decision processes~\cite{baierbook} to linear stochastic systems. Intuitively, the algorithm finds the set $\X_{init}$ using two fixed-point computations. The first one computes the set of all states that can reach the given target polytopes with non-zero probability. As a result, the remaining states of the state space $\X$ have zero probability of reaching the target polytopes. The second fixed-point computation finds the attractor of this set, \ie all states that have non-zero probability, under each control input from $\U$, of ever transiting to a state from which the target polytopes cannot be reached. Finally, the complement of the attractor is the desired set $\X_{init}$.

Note that Alg.~\ref{alg:polytopic} operates directly on the linear stochastic system. It performs polytopic operations only, and it involves neither refinement nor building a product with an automaton. Therefore, it performs considerably faster than the abstraction-refinement algorithm from Sec.~\ref{sec:solution}, as shown in Tab.~\ref{tab:stats}. However, it has two serious drawbacks. 

\begin{table}[t]
\caption{Statistical comparison of the specialized algorithm for reachability and our approach. 
}\label{tab:stats}
\begin{center}
{
\renewcommand{\arraystretch}{1.4}
\begin{tabular}{r l}
\hline
\multicolumn{2}{c}{Algorithm \ref{alg:polytopic}}\\
\hline
1st fixed point: & in 7 iterations, in <1 sec.\\
2nd fixed point: & in 1 iteration, in <1 sec.\\ 
\hline
\multicolumn{2}{c}{Abstraction-refinement from Sec.~\ref{sec:solution}}\\
\hline
Initial partition: & in 3 sec. \\
& game: 13 states, 27 actions\\
1st iteration: & in 7 min. \\
& game: 85 states, 712 actions\\
2nd iteration: & in 19 min. \\
& game: 131 states, 1262 actions\\
3rd iteration: & in 56 min. \\
& game: 250 states, 2724 actions\\
\hline
\end{tabular}
}
\end{center}
\end{table}

Firstly, Alg.~\ref{alg:polytopic} computes the set of satisfying initial states of the system, but no satisfying strategy. In extreme cases, every state may use a different control input in order to reach polytopes 
computed during the fixed-point computations, as in Ex.~\ref{ex:infinite}. 
In order to extract a finite satisfying strategy (if there is one), these polytopes have to be partitioned to smaller polytopes so that a fixed input can be used in all states of the new polytope. This partitioning is exactly the refinement procedure that our method performs when applied to reachability. Note that simpler refinement methods such as constructing only the NTS $\N_{\{\X_i\}}$, which is the first step of the abstraction in Sec.~\ref{subsec:abstraction}, are not sufficient, see App.~\ref{app:algabsref}. 
The whole 2\nicefrac{1}{2}-player game abstraction presented in Sec.~\ref{subsec:abstraction} is necessary for ensuring the correctness of the strategy. 

Secondly, Alg.~\ref{alg:polytopic} cannot be used for more complex properties than reachability. For more complex formulas, the product of the game with the automaton for the formula needs to be considered, since 
a winning strategy may require memory and pure polytopic methods can only provide memoryless strategies. In contrast, our abstraction-refinement approach designed in Sec.~\ref{sec:solution} works for general GR(1) properties. Moreover, it could easily be extended to the whole LTL at the cost of a higher complexity. 

\begin{figure*}[t]
{
\renewcommand{\arraystretch}{1.4}
\begin{tabular}{c c c c}
\hline
\multicolumn{4}{c}{Algorithm~\ref{alg:polytopic}}\\
Initial partition & First fixed-point alg. & Second fixed-point alg. & Final result\\
\includegraphics[scale=0.2]{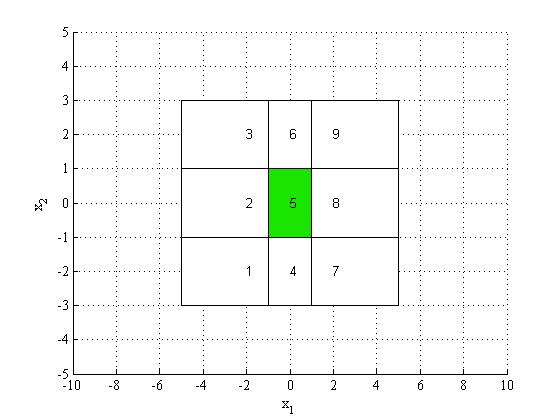} &
\includegraphics[scale=0.2]{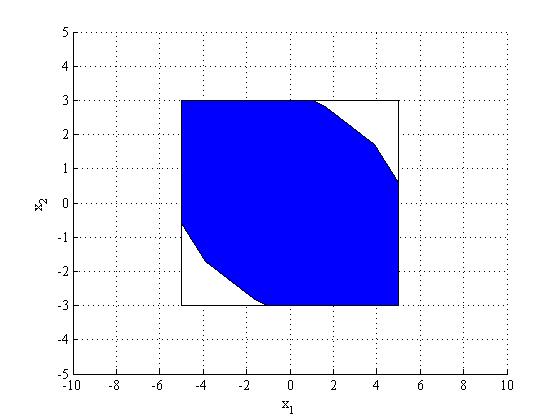} &
\includegraphics[scale=0.2]{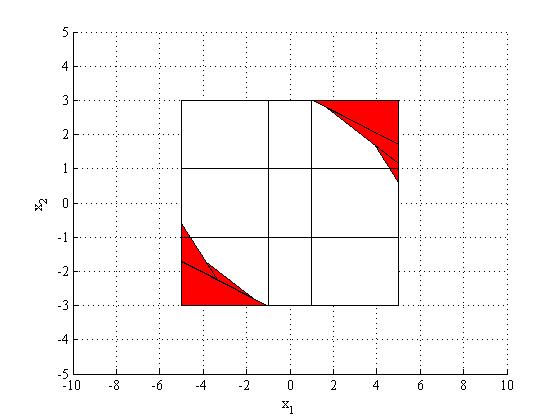} &
\includegraphics[scale=0.2]{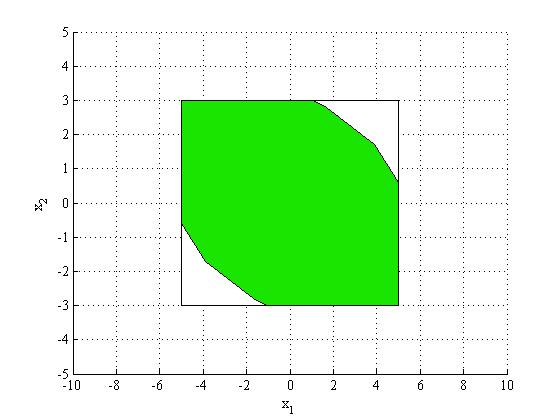} \\
\hline
\multicolumn{4}{c}{Abstraction-refinement from Sec.~\ref{sec:solution}}\\
Initial partition & First iteration & Second iteration & Third iteration\\
\includegraphics[scale=0.2]{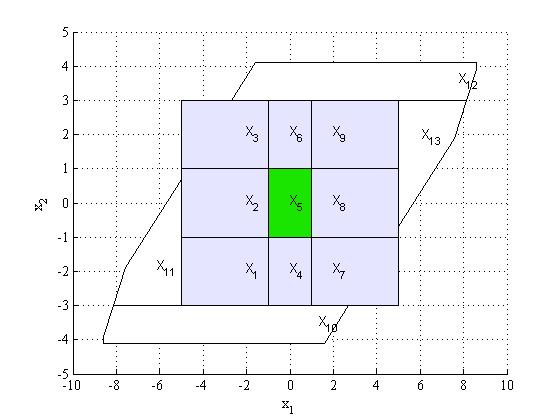} &
\includegraphics[scale=0.2]{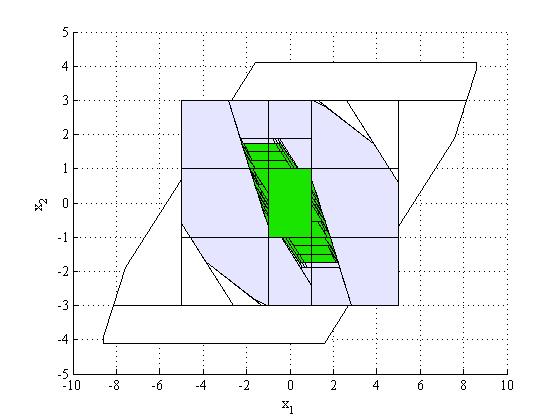} &
\includegraphics[scale=0.2]{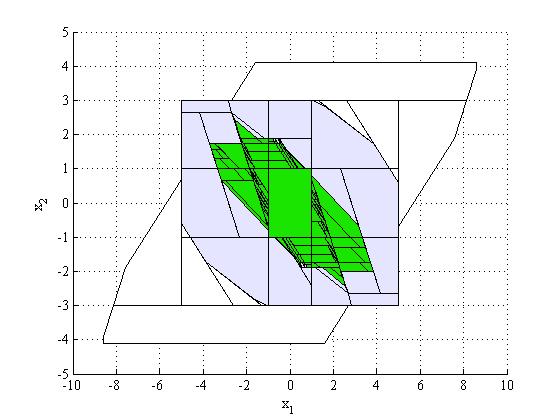} &
\includegraphics[scale=0.2]{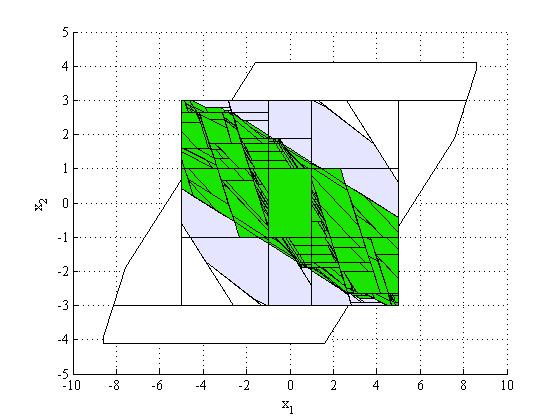}\\
\hline
\end{tabular}
}
\caption{Simulation of Alg.~\ref{alg:polytopic} and the abstraction-refinement algorithm. For Alg.~\ref{alg:polytopic}, we first depict the initial partition of $\X$ according to $\Pi$, with the polytope we aim to reach in green. The following columns show the two fixed point sets, in blue and red, respectively. The last column shows the resulting set $\X_{init}$. For the abstraction-refinement from Sec.~\ref{sec:solution}, we depict the results of for the initial partition and the next three iterations, where polytopes from sets $S_{yes},S_?,S_{no}$ are shown in green, light blue and white, respectively.}\label{fig:results}
\end{figure*}

We implemented both algorithms in Matlab, on a dual-core Intel i7 processor with 8 GB of RAM. The results are summarized in Fig.~\ref{fig:results} and Tab.~\ref{tab:stats}. For Alg.~\ref{alg:polytopic}, the set $\X_{init}$ was computed fast but it is a single polytope that does not provide any information about the satisfying strategies. For the abstraction-refinement algorithm, we computed the initial game and the following three iterations. Unlike for Alg.~\ref{alg:polytopic}, in every iteration, a satisfying strategy for a state $x$ in the partial solution is constructed as described in the proof of Prop.~\ref{prop:syes}.



\section{Conclusion and future work}\label{sec:futurework}

In this work, we considered the problem of computing the set of initial states of a linear stochastic system such that there exists a control strategy to ensure a GR(1) specification over states of the system. The solution is based on iterative abstraction-refinement using a 2\nicefrac{1}{2}-player game. Every iteration of the algorithm provides a partial solution given as a set of satisfying initial states with the satisfying strategies, and a set of non-satisfying initial states. 

While the algorithm guarantees progress and soundness in every iteration, it's complexity calls for more efficient implementation. The analyzed case study with a reachability property indicates that the current design would be too complex to deal with more complex properties such as persistent surveillance. In our future work, we aim to design efficient heuristic refinements that minimize the overall computation time for both reachability and general GR(1).

\noindent\textbf{Acknowledgement}
We thank Ebru Aydin Gol for useful discussions on the abstraction and polytopic computations.

{
\bibliographystyle{plain}
\bibliography{writeup}

\begin{thebibliography}{10}

\bibitem{abateTAC11}
A.~Abate, A.~D'Innocenzo, and M.D. Di~Benedetto.
\newblock {Approximate Abstractions of Stochastic Hybrid Systems}.
\newblock {\em IEEE TAC}, 56(11):2688--2694, 2011.

\bibitem{baierbook}
C.~Baier and J.P. Katoen.
\newblock {\em Principles of model checking}.
\newblock The MIT Press, 2008.

\bibitem{buchiL69}
J.~R. B{\"u}chi and L.~H. Landweber.
\newblock Solving sequential conditions by finite-state strategies.
\newblock {\em Trans. Amer. Math. Soc.}, 138:367--378, 1969.

\bibitem{ChaThesis}
K.~Chatterjee.
\newblock {\em Stochastic $\omega$-Regular Games}.
\newblock PhD thesis, UC Berkeley, 2007.

\bibitem{CCD14}
K.~Chatterjee, M.~Chmelik, and P.~Daca.
\newblock {CEGAR} for qualitative analysis of probabilistic systems.
\newblock In {\em CAV}, pages 473--490, 2014.

\bibitem{CdAH11}
K.~Chatterjee, L.~de~Alfaro, and T.~A. Henzinger.
\newblock Qualitative concurrent parity games.
\newblock {\em ACM TOCL}, 12(4), 2011.

\bibitem{twoandhalfplayergames}
K.~Chatterjee and T.~A. Henzinger.
\newblock A survey of stochastic $\omega$-regular games.
\newblock {\em JCSS}, 78(2):394--413, 2012.

\bibitem{CHJM05}
K.~Chatterjee, T.~A. Henzinger, R.~Jhala, and R.~Majumdar.
\newblock Counterexample-guided planning.
\newblock In {\em {UAI}}, 2005.

\bibitem{church57}
A.~Church.
\newblock Applications of recursive arithmetic to the problem of circuit
  synthesis.
\newblock {\em Summaries of the Summer Institute of Symbolic Logic, Cornell
  Univ.}, I:3--50, 1957.

\bibitem{Clarke99}
E.~M.~M. Clarke, D.~Peled, and O.~Grumberg.
\newblock {\em Model checking}.
\newblock MIT Press, 1999.

\bibitem{Girard:2010}
A.~Girard.
\newblock Synthesis using approximately bisimilar abstractions: state-feedback
  controllers for safety specifications.
\newblock In {\em Proc. of HSCC}, pages 111--120, 2010.

\bibitem{AyDiLaBe-CDC-2012}
E.~A. Gol, X.~Ding, M.~Lazar, and C.~Belta.
\newblock {Finite Bisimulations for Switched Linear Systems}.
\newblock In {\em Proc. of IEEE CDC}, pages 7632--7637, 2012.

\bibitem{ebruhscc12}
E.~A. Gol, M.~Lazar, and C.~Belta.
\newblock Language-guided controller synthesis for discrete-time linear
  systems.
\newblock In {\em Proc. of HSCC}, pages 95--104, 2012.

\bibitem{mengicra13}
M.~Guo, K.~H. Johansson, and D.~V. Dimarogonas.
\newblock {Revising motion planning under Linear Temporal Logic specifications
  in partially known workspaces}.
\newblock In {\em Proc. of IEEE ICRA}, pages 5025--5032, 2013.

\bibitem{norman_qest11}
E.~M. Hahn, G.~Norman, D.~Parker, B.~Wachter, and L.~Zhang.
\newblock {Game-based Abstraction and Controller Synthesis for Probabilistic
  Hybrid Systems}.
\newblock In {\em Proc. of QEST}, pages 69--78, 2011.

\bibitem{HJM03}
T.~A. Henzinger, R.~Jhala, and R.~Majumdar.
\newblock Counterexample-guided control.
\newblock In {\em {ICALP}}, 2003.

\bibitem{Agung:ACC:2006}
A.~A. Julius, A.~Girard, and G.~J. Pappas.
\newblock Approximate bisimulation for a class of stochastic hybrid systems.
\newblock In {\em Proc. of IEEE ACC}, pages 4724--4729, 2006.

\bibitem{norman_fmsd10}
M.~Kattenbelt, M.~Kwiatkowska, G.~Norman, and D.~Parker.
\newblock A game-based abstraction-refinement framework for {M}arkov decision
  processes.
\newblock {\em Formal Methods in System Design}, 36(3):246--280, 2010.

\bibitem{mortezacdc12}
M.~Lahijanian, S.~B. Andersson, and C.~Belta.
\newblock {Approximate Markovian abstractions for linear stochastic systems}.
\newblock In {\em Proc. of IEE CDC}, pages 5966--5971, 2012.

\bibitem{ozay_cdc14}
P.~Nilsson and N.~Ozay.
\newblock Incremental synthesis of switching protocols via abstraction
  refinement.
\newblock In {\em Proc. of IEEE CDC}, pages 6246--6253, 2014.

\bibitem{gr1def}
N.~Piterman, A.~Pnueli, and Y.~Sa'ar.
\newblock {Synthesis of Reactive(1) Designs}.
\newblock In {\em Proc. of VMCAI}, volume 3855 of {\em LNCS}, pages 364--380.
  2006.

\bibitem{PR89}
A.~Pnueli and R.~Rosner.
\newblock {On the Synthesis of a Reactive Module}.
\newblock In {\em Proc. of POPL}, pages 179--190, 1989.

\bibitem{majacdc13}
M.~Svorenova, I.~Cerna, and C.~Belta.
\newblock {Optimal Control of MDPs with Temporal Logic Constraints}.
\newblock In {\em Proc. of IEEE CDC}, pages 3938--3943, 2013.

\bibitem{SvCeBe-ACC-2013}
M.~Svorenova, I.~Cerna, and C.~Belta.
\newblock {Optimal Receding Horizon Control for Finite Deterministic Systems
  with Temporal Logic Constraints}.
\newblock In {\em Proc. of IEEE ACC}, pages 4399--4404, 2013.

\bibitem{TP03}
P.~Tabuada and G.~Pappas.
\newblock Model checking {LTL} over controllable linear systems is decidable.
\newblock In {\em Proc. of HSCC}, volume 2623 of {\em LNCS}, pages 498--513.
  2003.

\bibitem{ufukhscc12}
U.~Topcu, N.~Ozay, J.~Liu, and R.~M. Murray.
\newblock { Synthesizing Robust Discrete Controllers Under Modeling
  Uncertainty}.
\newblock In {\em Proc. of ACM HSCC}, pages 85--94, 2012.

\bibitem{eric_cdc12}
E.~M. Wolff, U.~Topcu, and R.~M. Murray.
\newblock Robust control of uncertain markov decision processes with temporal
  logic specifications.
\newblock In {\em Proc. of IEEE CDC}, pages 3372--3379, 2012.

\bibitem{boyan}
B.~Yordanov, J.~Tumova, I.~Cerna, J.~Barnat, and C.~Belta.
\newblock {Temporal Logic Control of Discrete-Time Piecewise Affine Systems}.
\newblock {\em IEEE TAC}, 57:1491--1504, 2012.

\end{thebibliography}
}


\appendix

\section{Solving a 2\nicefrac{1}{2}-player game}\label{app:gamesalg}

\newcommand{\pre}{{\textsf{Pre}_1}}
\newcommand{\apre}{{\textsf{Pre}_2}}
\newcommand{\cpre}{{\textsf{Pre}_3}}
\newcommand{\wb}{\overline}
\newcommand{\Condpre}{\mathsf{C}_1}
\newcommand{\Condapre}{\mathsf{C}_2}
\newcommand{\Condcpre}{\mathsf{C}_3}

\algdef{SE}[DOWHILE]{Do}{doWhile}{\algorithmicdo}[1]{\algorithmicwhile\ #1}%


Here we present an algorithm to solve the almost-sure winning problem for a 2\nicefrac{1}{2}-player game $\G=(S_1,S_2,\act,\delta)$ with a B\"uchi implication condition $(E,F)$, where $E,F\subseteq S$. The optimal solution is a rather involved, quadratic time algorithm that can be found in~\cite{ChaThesis}. In this work, we use a more intuitive, cubic time algorithm presented in Alg.~\ref{alg:almost}, whose correctness follows from~\cite{CdAH11}. The algorithm is a simple iterative fixed-point algorithm that uses three types of local predecessor operator over the set of states of the game. 


Consider sets $X,Y,Z$ such that $Y \subseteq Z \subseteq X \subseteq S$. 
Given a state $s\in S$ and an action $a\in \act$, we denote by $\dest(s,a)= \supp(\delta(s,a))$ the set of possible successors of the state and the action. We define conditions on state action pairs as follows:
\begin{align*}
\Condpre(X) =  \{ (s,a) \mid & \dest(s,a) \subseteq X\}, \\
\Condapre(X,Y)  =  \{ (s,a) \mid & \dest(s,a) \subseteq X \text{ and} \\ 
& \dest(s,a) \cap Y \neq \emptyset \},\\
\Condcpre(Z,X,Y) = \{ (s,a) \mid & (\dest(s,a) \subseteq Z) \text{ or}\\
& (\dest(s,a) \subseteq X \text{ and}\\ 
& \dest(s,a) \cap Y \neq \emptyset) \}.
\end{align*}
The first condition ensures that given the state and action the next state is in $U$ with probability~1, the second condition ensures that the next state is in $X$ with probability~1 and in $Y$ with positive probability. The third condition is the disjunction of the first two. The three predecessor operators are defined as the set of Player~1, or Player~2 states, where there exists, or for all, respectively, actions, the condition for the  predecessor operator is satisfied. The three respective predecessor operators, namely, $\pre,\apre,$ and $\cpre$ are defined as follows:

\small
\vspace*{-\baselineskip}
\begin{align*}
\pre(X)  =& \{ s \in S_1 \mid \exists a \in \act. \ (s,a) \in \Condpre(X) \} \  \cup \\
& \{ s \in S_2 \mid \forall a \in \act. \ (s,a) \in \Condpre(X)\},\\
\apre(X,Y)  =&  \{ s \in S_1 \mid \exists a \in \act. \ (s,a) \in \Condapre(X,Y) \} \  \cup \\
& \{ s \in S_2 \mid \forall a \in \act. \ (s,a) \in \Condapre(X,Y)\},\\
\cpre(Z,X,Y)  =&  \{ s \in S_1 \mid \exists a \in \act. \ (s,a) \in \Condcpre(Z,X,Y) \} \  \cup \\
& \{ s \in S_2 \mid \forall a \in \act. \ (s,a) \in \Condcpre(Z,X,Y)\}. 
\end{align*}
\normalsize


\begin{algorithm}[t]
  \caption{Algorithm for $\almost^{\G}(\varphi)$}
  \label{alg:almost}
\begin{algorithmic}
\State \textbf{Input:} game $\G$, acc. condition $(E,F)$, $D = S \setminus (E \cup F)$;
\State Set: $X,Y,Z,\wb{X},\wb{Y},\wb{Z};$ 
\State $\wb{X} \gets S; \quad \wb{Z} \gets S; \quad \wb{Y} \gets \emptyset;$
\Comment Initialization
\Do
 \State $X \gets \wb{X}$
 \Do
  \State $Y \gets \wb{Y};$
  \Do
  \State $Z \gets \wb{Z};$
  \State $\wb{Z} \gets (F \cap \pre(X)) \cup (E \cap \apre(X,Y)) \cup$ \\ \hspace{7em} $ (D \cap \cpre(Z,X,Y);$
  \doWhile{$Z\neq \wb{Z}$}
  \State $\wb{Y}\gets Z;$
  \State $\wb{Z} \gets S;$
 \doWhile{$Y\neq \wb{Y}$}
 \State $\wb{X} \gets Y$
 \State $\wb{Y} \gets \emptyset;$
\doWhile{$X\neq \wb{X}$}
\State \textbf{return:} $X$
\end{algorithmic}
\end{algorithm}

\section{Polytopic operators}\label{app:polytopicop}
In this section, we describe in detail the computation of all polytopic operators introduced in Sec.~\ref{sec:solution} and used in our solution to Problem~\ref{pf}.


\subsection{Action polytopes}\label{subapp:UiJ}

First, we describe how to compute the action polytopes $\U_i^J$ for every polytope $\X_i\in \{\X_i\}_{i\in I}$, formally defined in Eq.~(\ref{eq:UiJ}).

For a polytope $\X'\subset \mathbb{R}^N$, we use $\U^{\X_i\to \X'}$ to denote the set of all control inputs from $\U$ under which the system $\T$ can evolve from a state in $\X_i$ to a state in $\X'$ with non-zero probability, \ie
\begin{align}\label{eq:Ufromto}
\U^{\X_i\to \X'} = \{u\in \U \mid & \post(\X_i,u)\cap \X' \text{ is non-empty}\}.
\end{align}
The following proposition states that $\U^{\X_i\to \X'}$ can be computed from the V-representations of $\X_i,\X'$ and $\W$.

\begin{proposition}\label{prop:Ufromto}
Let $H,K$ be the matrices from the H-representation of the following polytope:
\begin{equation}\label{eq:Ufromtohelp1}
\{ y\in \mathbb{R}^N \mid \exists x\in \X_i,\exists w\in \W:\, Ax+y+w\in \X' \},
\end{equation}
which can be computed as the convex hull 
\begin{align}
\hull(\{ & v_{\X'} - (Av_{\X_i} + v_{\W}) \mid v_{\X'}\in \vertices(\X'),\nonumber \\
&\quad  v_{\X_i}\in \vertices(\X_i), v_{\W}\in \vertices(\W) \}).\label{eq:Ufromtohelp2}
\end{align}
Then the set $\U^{\X_i\to \X'}$ defined in Eq.~(\ref{eq:Ufromto}) is the polytope with the following H-representation:
\begin{equation}\label{eq:Ufromtocomp}
\U^{\X_i\to \X'} = \{u\in \U\mid HBu\leq K\}.
\end{equation}
\end{proposition}

\begin{proof}
To fact that the set in Eq.~(\ref{eq:Ufromtohelp1}) is a polytope with the V-representation given in Eq.~(\ref{eq:Ufromtohelp2}) can be easily shown as follows. Let $y\in \mathbb{R}^N$ be such that there exist $x\in \X_i,w\in \W, x'\in \X'$ for which $Ax+y+w=x'$, \ie $y=x'-(Ax+w)$. By representing $x',x$ and $w$ as an affine combination of the respective vertices in $\vertices(\X'), \vertices(\X_i)$ and $\vertices(\W)$, we obtain the V-representation in Eq.~(\ref{eq:Ufromtohelp2}). Next, let $H, K$ be the matrices from the H-representation of the set in Eq.~(\ref{eq:Ufromtohelp1}). Then the definition of set $\U^{\X_i\to \X'}$ in Eq.~(\ref{eq:Ufromto}) can be written as
\begin{align*}
\U^{\X_i\to \X'} = \{u\in \U\mid & \exists x\in \X_i, \exists w\in \W:\\ 
& \qquad Ax+Bu+w\in \X'\},
\end{align*}
that leads to H-representation in Eq.~(\ref{eq:Ufromtocomp}). 
\end{proof}

\begin{corollary}
Let $J\subseteq I\cup I_{\out}$. The set $\U_i^J$ from Eq.~(\ref{eq:UiJ}) can be computed as follows:
\begin{equation}\label{eq:UiJcomp}
\U_i^J = \bigcap \limits_{j\in J} \U^{\X_i\to \X_j} \backslash \bigcup \limits_{j'\not \in J} \U^{\X_i\to \X_{j'}}.
\end{equation}
\end{corollary}

\begin{proof}
Follows directly from Eq.~(\ref{eq:UiJ}) and~(\ref{eq:Ufromto}).
\end{proof}

Note that $\U_i^J$ is generally not a polytope but can be represented as a finite union of polytopes. 


\subsection{Posterior}

The posterior operator $\post(\X',\U')$, formally defined in Tab.~\ref{tab:polytopeops}, can be easily computed using Min\-kowski sum as
\begin{align*}
\post(\X'\U') &= A\X'+B\U'+\W\\
&= \hull(\{ A v_{\X'} + Bv_{\U'} + v_{\W} \mid v_{\X'}\in \vertices(\X'), \\
&\quad \quad  v_{\U'}\in \vertices(\U'), v_{\W}\in \vertices(\W) \}).
\end{align*}


\subsection{Predecessor}

The predecessor operator $\preo(\X',\U',\{\X_j\}_{j\in J})$, formally defined in Tab.~\ref{tab:polytopeops}, can be computed as follows. First, note that 
\begin{equation*}
\preo(\X',\U',\{\X_j\}_{j\in J}) = \bigcup \limits_{j\in J} \preo(\X',\U',\X_j).
\end{equation*}

\begin{proposition}\label{prop:preo}
Let $H,K$ be the matrices from the H-representation of the following polytope:
\begin{equation*}\label{eq:preohelp1}
\{ y\in \mathbb{R}^N \mid \exists u\in \U',\exists w\in \W:\, y+Bu+w\in \X_j \},
\end{equation*}
which can be computed as the convex hull 
\begin{align*}
\hull(\{ & v_{\X_j} - (Bv_{\U'} + v_{\W}) \mid v_{\X_j}\in \vertices(\X_j),\nonumber \\
&\quad  v_{\U'}\in \vertices(\U'), v_{\W}\in \vertices(\W) \}).\label{eq:preohelp2}
\end{align*}
Then the set $\preo(\X',\U',\X_j)$ is the polytope with the following H-representation:
\begin{equation*}\label{eq:preocomp}
\preo(\X',\U',\X_j) = \{x\in \X'\mid HAx\leq K\}.
\end{equation*}
\end{proposition}

\begin{proof}
The proof is analogous to the one of Prop.~\ref{prop:Ufromto}.
\end{proof}


\subsection{Robust and precise predecessor}\label{subapp:prerp}

From definitions of the robust and precise predecessor operators in Tab.~\ref{tab:polytopeops} it follows that
\begin{align*}
& \prer(\X',\U',\{\X_j\}_{j\in J}) = \\
&\quad \quad \bigcup \limits_{J'\subseteq J,J'\neq \emptyset} \prep(\X',\U',\{\X_j\}_{j\in J'}).
\end{align*}
Below we describe the computation of the precise predecessor $\prep(\X',\U',\{\X_j\}_{j\in J'})$ for any $J'\subseteq J$.

Let $\Z$ denote the polytope, or finite union of polytopes, $\Z=A\X'+B\U'$, where $+$ denotes the Minkowski sum. For a polytope $\mathcal{P}\subset \mathbb{R}^N$, we define set
\begin{equation}
\Z(\mathcal{P}) = \{z\in \Z\mid (z + \W)\cap \mathcal{P} \text{ is non-empty}\}. \label{eq:ZiJXprime}
\end{equation}
For a set of polytopes $\{\mathcal{P}\}$, $\Z(\{\mathcal{P}\})$ can be computed as the union of all $\Z(\mathcal{P})$ for every polytope $\mathcal{P}$ in the set $\{\mathcal{P}\}$.

\begin{proposition}
The set from Eq.~(\ref{eq:ZiJXprime}) is the following polytope, or finite union of polytopes:
\begin{align}
\Z(\mathcal{P}) = hull(\{v_{\mathcal{P}}-v_{\W}\mid & v_{\mathcal{P}}\in \vertices(\mathcal{P}),\nonumber \\
& v_{\W}\in \vertices(\W)\}) \cap \Z. \label{eq:ZiJjcomp}
\end{align}
\end{proposition}

\begin{proof}
The proof is carried out in a similar way as the first part of proof of Prop.~\ref{prop:Ufromto}.
\end{proof}

For $J'\subseteq J$, we use $\Z(J')$ to denote the set
\begin{equation}\label{eq:ZiJJprime}
\Z(J') = \bigcap \limits_{j\in J'} \Z(\X_j) \backslash \big( \bigcup \limits_{j\in J\backslash J'} \Z(\X_j) \cup \Z(\X_{\neg J})\big),
\end{equation}
where $\Z(\X_{\neg J}) = \Z((\X\cup \X_{\out})\backslash \bigcup \limits_{j\in J}\X_j)$.

\begin{proposition}
Let $\U'=\{\U_{l_1}\}_{l_1\in L_1}$, $J\subseteq J'$ and let $\Z(J') = \{\Z_{l_2}\}_{l_2\in L_2}$. Then the precise predecessor can be written as
\begin{align}
\prep(\X',\U',\{\X_j\}_{j\in J'}) = & \bigcup \limits_{l_1\in L_1} \bigcup \limits_{l_2\in L_2} \{ x\in \X' \mid  \nonumber \\
& \exists u\in \U_{l_1}: Ax+Bu\in \Z_{l_2} \}.\label{eq:prercomp}
\end{align}
Let $l_1\in L_1,l_2\in L_2$ and let $H,K$ be the matrices from the H-representation of the following polytope:
\begin{equation}
\{y\in \mathbb{R}^N\mid \exists u\in \U_{l_1}: y+Bu\in \Z_{l_2} \},
\end{equation}
which can be computed as the convex hull
\begin{equation}
\hull ( \{ v_{\Z_{l_2}} - Bv_{\U_{l_1}}\mid v_{\Z_{l_2}}\in \vertices({\Z_{l_2}}), v_{\U_{l_1}}\in \vertices({\U_{l_1}}) \} ).
\end{equation}
Then the set on the right-hand site of Eq.~(\ref{eq:prercomp}), for $l_1,l_2$, is a polytope with the following H-representation:
\begin{equation}
\{ x\in \X'\mid HAx\leq K \}.
\end{equation}
\end{proposition}

\begin{proof}
From the definition of the set $\Z(J')$ in Eq.~(\ref{eq:ZiJJprime}), $z\in \Z(J')$ iff $z+\W$ intersects all $\X_j$ for $j\in J'$ and $z+\W\subseteq \bigcup \limits_{j\in J'} \X_j$. Moreover, every $z\in \Z$ can be written as $z=Ax+Bu$ and therefore $z+\W=\post(x,u)$. This proves Eq.~(\ref{eq:prercomp}). The rest of the proof is carried out in a way similar to the proof of Prop.~\ref{prop:Ufromto}.
\end{proof}


\begin{algorithm}[t]
  \caption{Computing the set $\X_{init}\subseteq \X$ of states from which a set of polytopes $\{\X_j\}_{j\in J}$ in $\X$ can be reached with probability 1, using abstraction to a NTS. 
  }
  \label{alg:absref}
\begin{algorithmic}
\State \textbf{Input:} linear stochastic system $\T$, partition $\{\X_i\}_{i\in I}$ of state space $\X$, subset $J\subseteq I$
\State $\X_{>0}\gets \emptyset$
\State $\X_{>0}' \gets \{\X_j\}_{j\in J}$
\While{$\X_{>0}\neq \X_{>0}'$}
  \State $\X_{>0}\gets \X_{>0}'$;
  \State construct NTS $\N_{\{\X_i\}}$ for current partition (Sec.~\ref{subsec:abstraction})  
  \For{every state $\X_i\not \subseteq \X_{>0}$}
    \State refine $\X_i$ according to $\preo(\X_i,\U,\X_{>0})$;
    \State $\X_{>0}'\gets \X_{>0}' \cup \preo(\X_i,\U,\X_{>0})$;
  \EndFor  
\EndWhile
\State $\X_{=0,attr}\gets \X_\out \cup \X$
\State $\X_{=0,attr}' \gets \X_\out \cup \X\backslash \X_{>0}$
\While{$\X_{=0,attr}\neq \X_{=0,attr}'$}
  \State $\X_{=0,attr}\gets \X_{=0,attr}'$;
  \State construct NTS $\N_{\{\X_i\}}$ for current partition (Sec.~\ref{subsec:abstraction})  
  \For{every state $\X_i$ s.t. all actions lead to $\X_{=0,attr}$}
    \State refine $\X_i$ according to $\attr(\X_i,\U,\X_{=0,attr})$;
    \State $\X_{=0,attr}'\gets \X_{=0,attr}' \cup \attr(\X_i,\U,\X_{=0,attr})$;
  \EndFor  
\EndWhile
\State $\X_{=1}\gets \X\backslash \X_{=0,attr};$
\State \textbf{return:} $\X_{=1}$
\end{algorithmic}
\end{algorithm}

\begin{figure*}[t]
{
\renewcommand{\arraystretch}{1.4}
\begin{tabular}{c c c c}
\hline
\multicolumn{4}{c}{Algorithm~\ref{alg:absref}}\\
Initial partition & First fixed-point alg. & Second fixed-point alg. & Final result\\
\includegraphics[scale=0.2]{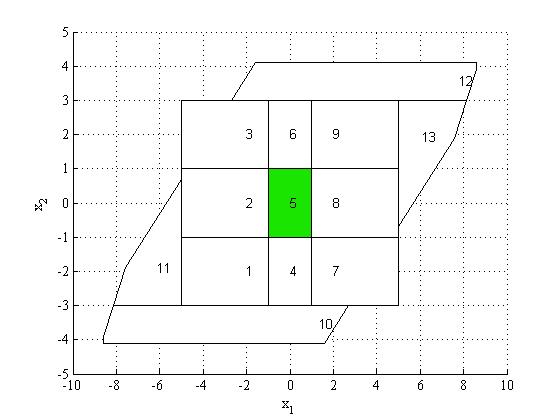} &
\includegraphics[scale=0.2]{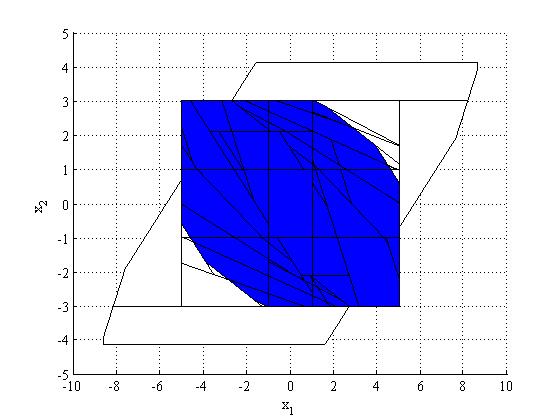} &
\includegraphics[scale=0.2]{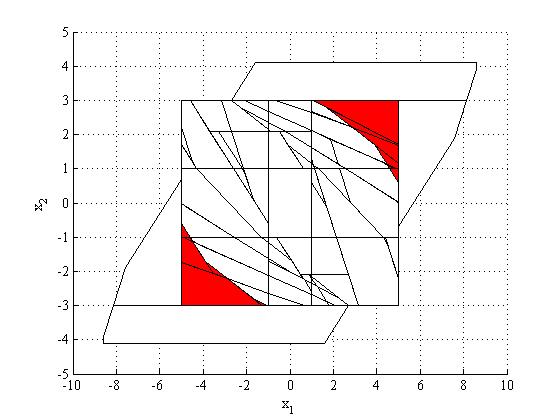} &
\includegraphics[scale=0.2]{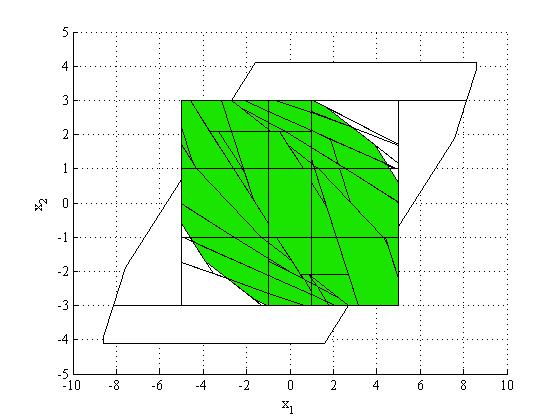} \\
\hline
\end{tabular}
}
\caption{Results obtained from simulation of Alg.~\ref{alg:absref} for the case study from Sec.~\ref{sec:casestudy}. In the first column, we depict the initial partition of $\X$ according to $\Pi$ and polytopes from $\X_\out$, with the polytope we aim to reach in green. The following columns show the fixed point sets, in blue and red, respectively, together with the obtained partition. The last column shows the resulting set $\X_{init}$. }\label{fig:results2}
\end{figure*}

\subsection{Attractor}

The attractor operator $\attr(\X',\U',\{\X_j\}_{j\in J})$ from Tab.~\ref{tab:polytopeops} can be computed using the robust predecessor operator, since it holds that
\begin{align*}
& \attr(\X',\U',\{\X_j\}_{j\in J}) =\\
&= \{x\in \X'\mid \forall u\in \U':\, \post(x,u) \cap \bigcup \limits_{j\in J}\X_j \text{ is non-empty}\}\\
&= \X' \backslash \{x\in \X'\mid \exists u\in \U':\, \post(x,u)\subseteq (\X\cup \X_{\out})\backslash \bigcup \limits_{j\in J}\X_j\}\\
&= \X'\backslash \prer(\X',\U',(\X\cup \X_{\out})\backslash \bigcup \limits_{j\in J}\X_j).
\end{align*}


\subsection{Robust attractor}

The robust attractor operator $\attrr(\X',\U',\{\X_j\}_{j\in J})$ from Tab.~\ref{tab:polytopeops} can be computed using the predecessor operator, since it holds that
\begin{align*}
& \attrr(\X',\U',\{\X_j\}_{j\in J}) =\\
&= \{x\in \X'\mid \forall u\in \U':\, \post(x,u) \subseteq \bigcup \limits_{j\in J}\X_j\}\\
&= \X' \backslash \{x\in \X'\mid \exists u\in \U':\, \post(x,u)\cap (\X\cup \X_{\out})\backslash \bigcup \limits_{j\in J}\X_j \\
& \quad \quad \quad \quad \quad \quad \quad \text{ is non-empty}\}\\
&= \X'\backslash \preo(\X',\U',(\X\cup \X_{\out})\backslash \bigcup \limits_{j\in J}\X_j).
\end{align*}

\section{Approach comparison}\label{app:algabsref}

Here, we can compare our abstraction-refinement approach from Sec.~\ref{sec:solution} to the algorithm for reachability presented in Alg.~\ref{alg:absref}. Alg.~\ref{alg:absref} combines the simple approach from Alg.~\ref{alg:polytopic} that uses only polytopic operations with the abstraction-refinement method. In every iteration, we build the non-deterministic transition system $\N_{\{\X_i\}}$, which is the first step of the abstraction in Sec.~\ref{subsec:abstraction}. The partition $\{\X_i\}_{i\in I}$ is then iteratively refined using the two fixed-point algorithms as in Alg.~\ref{alg:polytopic}. 

Just like in Alg.~\ref{alg:polytopic}, Alg.~\ref{alg:absref} operates directly on the linear stochastic system. It uses polytopic operators only and does not build a product with any automaton. Therefore, it performs faster than our approach, as demonstrated below. It however suffers from the same two serious drawbacks as the polytopic method. Firstly, it finds the set of satisfying initial states of the system, but no satisfying strategy. However, in comparison with Alg.~\ref{alg:polytopic}, it can provide at least a partial information on the properties of satisfying strategies. Namely, it specifies for every polytope of the resulting partition $\{\X_i\}$ of the state space $\X$ which control inputs cannot be used in any satisfying strategy. As we are interested in reachability property, these are the control inputs for which the corresponding non-deterministic transition leads from $\X_i$ outside of $\X_{init}$. Secondly, just like Alg.~\ref{alg:polytopic}, Alg.~\ref{alg:absref} cannot be used for more complex properties than reachability. As discussed in Sec.~\ref{sec:casestudy}, the product of the game with the automaton needs to be considered for more complex properties.

The results from simulations of Alg.~\ref{alg:absref} are presented in Fig.~\ref{fig:results2} and Tab.~\ref{tab:stats2}. The algorithm found fixed point sets for both fixed-point computation rather quickly, and in the same number of iterations as the polytopic algorithm in Alg.~\ref{alg:polytopic}, see Fig.~\ref{fig:results} and Tab.~\ref{tab:stats}. While Alg.~\ref{alg:absref} performs faster than our algorithm designed in Sec.~\ref{sec:solution}, it provides only partial information on the satisfying strategies, as discussed above. 


\begin{table}[t]
\caption{Statistics for the simulation of Alg.~\ref{alg:absref} for the case study from Sec.~\ref{sec:casestudy}.}\label{tab:stats2}
\begin{center}
{
\renewcommand{\arraystretch}{1.4}
\begin{tabular}{r l}
\hline
\multicolumn{2}{c}{Algorithm \ref{alg:absref}}\\
\hline
1st fixed point: & in 7 iterations, in 3 min. \\
& 1st NTS: 13 states, 27 actions\\
& 2nd NTS: 25 states 105 actions\\
& 3rd NTS: 45 states 289 actions\\
& 4th NTS: 63 states, 524 actions\\
& 5th NTS: 77 states, 745 actions\\
& 6th NTS: 88 states, 994 actions\\
& 7th NTS: 92 states, 1139 actions\\ 
2nd fixed point: & in 1 iteration, in 2 sec.\\
\hline
\end{tabular}
}
\end{center}
\end{table}



\end{document}